\documentclass[journal]{IEEEtran}
\usepackage{algorithm,algorithmic}

\usepackage{cite}

\ifCLASSINFOpdf
\else
\fi
\usepackage{mathtools}
\usepackage{amsmath}
\usepackage{amsthm}
\usepackage{amssymb}
\usepackage{caption}
\usepackage[pdftex]{graphicx}
\usepackage{breqn}
\usepackage{epstopdf}
\usepackage{setspace}
\usepackage{bm}
\usepackage{tikz}
\usetikzlibrary{matrix,decorations.pathreplacing}
\newcommand\norm[1]{\left\lVert#1\right\rVert}

\newtheorem{lemma}{Lemma}

\usepackage{arydshln}
\usepackage[flushleft]{threeparttable}
\usepackage{dblfloatfix}
\begin{document}
%
\title{Semidefinite Relaxation-Based PAPR-Aware Precoding for Massive MIMO-OFDM Systems}
%
%
%

\author{Miao~Yao,
		Matt~Carrick,
        Munawwar~M.~Sohul,
        Vuk~Marojevic,
        Cameron~D.~Patterson,
        and~Jeffrey~H.~Reed,~\IEEEmembership{Fellow,~IEEE}
\thanks{Miao Yao, Matt Carrick, Munawwar M. Sohul, Vuk Marojevic,  Cameron D. Patterson, and Jeffrey H. Reed are with the Department
of Electrical and Computer Engineering, Virginia Tech, Blacksburg,
VA, 24060 USA e-mail: (reedjh@vt.edu). This work is supported by NSF grant: CNS-1228903.}
}

%
%

\markboth{}%
{Shell \MakeLowercase{\textit{et al.}}: Bare Demo of IEEEtran.cls for Journals}
%



\makeatletter
\def\ps@IEEEtitlepagestyle{
  \def\@oddfoot{\mycopyrightnotice}
  \def\@evenfoot{}
}
\def\mycopyrightnotice{
  {\footnotesize
  \begin{minipage}{\textwidth}
  \centering
  Copyright~\copyright~2015 IEEE. Personal use of this material is permitted. However, permission to use this  \\ 
  material for any other purposes must be obtained from the IEEE by sending a request to pubs-permissions@ieee.org.
  \end{minipage}
  }
}

\maketitle

\begin{abstract}
Massive MIMO requires a large number of antennas and the same amount of power amplifiers (PAs), one per antenna. As opposed to 4G base stations, which could afford highly linear PAs, next-generation base stations will need to use inexpensive PAs, which have a limited region of linear amplification. One of the research challenges is effectively handling signals which have high peak-to-average power ratios (PAPRs), such as orthogonal frequency division multiplexing (OFDM). This paper introduces a PAPR-aware precoding scheme that exploits the excessive spatial degrees-of-freedom of large scale multiple-input multiple-output (MIMO) antenna systems. This typically requires finding a solution to a nonconvex optimization problem. Instead of relaxing the problem to minimize the peak power, we introduce a practical semidefinite relaxation (SDR) framework that enables accurately and efficiently approximating the theoretical PAPR-aware precoding performance for OFDM-based massive MIMO systems. The framework allows incorporating channel uncertainties and intercell coordination. Numerical results show that several orders of magnitude improvements can be achieved w.r.t. state of the art techniques, such as instantaneous power consumption reduction and multiuser interference cancellation. 
The proposed PAPR-aware precoding can be effectively handled along with the multi-cell signal processing by the centralized baseband processing platforms of next-generation radio access networks. Performance can be traded for the computing efficiency for other platforms.
\end{abstract}
\begin{IEEEkeywords}
Massive MIMO, OFDM, Semidefinite Relaxation, PAPR Reduction, Green Communications.
\end{IEEEkeywords}


\section{Introduction}

\IEEEPARstart{T}{he}  massive deployments of heterogeneous wireless networks and the emerging 5G new radio (NR) have motivated the demand for energy and spectrum efficiency to reduce operational costs. Orthogonal frequency division multiplexing (OFDM) has been adopted by the 4G long-term evolution (LTE) \cite{access2014physical}. OFDM suffers from the drawback of high peak-to-average power ratio (PAPR) at the transmitter \cite{ochiai2000performance,ren2003complementary,slimane2007reducing}. The high signal peaks which are generated by the constructive addition of different subcarriers lead to the signal excursions into the nonlinear region of the power amplifier (PA). Typically, to avoid nonlinear signal distortion, the input power to the PA is reduced or ''backed-off''. Operating at lower power levels reduces the power efficiency and increases the operational expenditures. This is not a realistic solution for 5G NR networks since the target energy efficiency improvement is 100x w.r.t. 4G deployments\cite{ITU21meeting}. Large scale multiple-input multiple-output (MIMO) (a.k.a. massive MIMO) systems have been proposed as the key enabler of 5G due to their ability to mitigate the multiuser interference (MUI) and improve link reliability and spectrum efficiency \cite{ngo2013energy}. Although highly-linear PAs are desirable, their elevated costs lead to capital expenditure that scales linearly with the number of antennas and is, thus, prohibitively expensive for massive MIMO base station (BS) deployments. Massive MIMO systems have the potential to reduce the PAPR on a symbol basis by exploiting the extra spatial degrees-of-freedom (DoFs) and hence adopt less expensive PAs   \cite{mohammed2013per,zhang2016per,pan2014constant,chen2016low,studer2013aware,cha2014generalized,bao2016efficient}.

A semidefinite relaxation (SDR)-based PAPR-aware precoding for massive MIMO-OFDM systems is proposed in this paper. We formulate the instantaneous transmit power minimization subject to PAPRs and MUI thresholds. It provides a mechanism for trading MUI, PAPR, and transmit power. This problem is a nonconvex quadratically constrained quadratic programming (QCQP) problem. We employ the SDR and then the rank reduction method to obtain an estimate of the QCQP. We leverage the SDR approach of \cite{huang2010rank,wiesel2005semidefinite,mobasher2007near,wang2011papr,zhang2014papr,ma2004semidefinite,ma2002quasi,goemans1995improved,sidiropoulos2006semidefinite} and apply a randomization scheme, as proposed in \cite{ma2002quasi,ma2004semidefinite,yao2017sustainable,sidiropoulos2006semidefinite} , to achieve PAPR-aware precoding solution. As a powerful optimization technique, SDR has recently been applied to solve a variety of nonconvex or NP-hard problems in communications. Reference \cite{huang2010rank} minimizes the transmit power of a beamforming problem while keeping the interference generated by other coexisting systems under a tolerable level using SDR. Maximum-likelihood detection in MIMO M-ary phase-shift keying (M-PSK) communication systems, is an NP-hard least squares search problem, which is also approximated by SDR \cite{ma2004semidefinite}. The application of SDR in MIMO detection was later extended to 16-quadrature amplitude modulation (16-QAM) \cite{wiesel2005semidefinite} and to higher-order QAM alphabets for general QAM constellations \cite{mobasher2007near}. In addition, SDR has been applied to PAPR reduction of the single antenna system in \cite{wang2011papr,zhang2014papr}. The basis of our solution is expanding the feasible solution sets yielding an attractive polynomial-time approximation. Improvements brought by the traditional PAPR reduction techniques trade effective transmission rate or spectral purity for improving power efficiency. This paper applies a PAPR-aware precoding scheme which takes advantage of excessive DoFs in massive MIMO to reduce PAPR while preserving transmission rate and out-of-band emissions.

Researchers have addressed the PAPR issue of OFDM. The proposed solutions introduce signal distortion or redundancy to improve efficiency \cite{ochiai2000performance,ochiai2002performance,ren2003complementary,slimane2007reducing,gazor2012tone,aggarwal2006minimizing,alavi2005papr,chen2010partial}. One of the simplest and most extensively used schemes is clipping and filtering \cite{ochiai2000performance,ochiai2002performance,ren2003complementary} which limits the PAPR below a threshold level, but causes both subcarrier inter-modulation and out-of-band radiation. The coding scheme proposed in \cite{slimane2007reducing} transforms the subcarriers into codewords with low PAPR, but reduces spectrum efficiency, especially when the number of subcarriers is large. Tone reservation \cite{gazor2012tone} and constellation error shaping \cite{aggarwal2006minimizing} avoid sending data on a small subset of subcarriers and extend outer constellation points to minimize the PAPR of the OFDM symbols, degrading spectrum efficiency. The selective mapping (SLM) scheme \cite{breiling2001slm} is realized by multiplying the input symbol sequence to select alternative input sequences; but this requires side information to recover the signal at the receiver. The partial transmit sequence (PTS) scheme \cite{alavi2005papr,chen2010partial} partitions the input symbol sequence into a variety of disjoint symbol subsequences and also needs additional side information. Unlike to the existing PAPR reduction approaches, the proposed PAPR-aware precoding solution does not sacrifice the transmission rate nor does it affect the spectral purity.

There are two main kinds of PAPR-aware precoding schemes exploiting excess DoFs in massive MIMO systems: constant envelope precoding and multiuser (MU) precoding. Constant envelope precoding \cite{mohammed2013per,zhang2016per,pan2014constant} achieves a relatively flat amplitude transmit signal envelope by using phase modulation. Moreover, the single-user constant envelope precoder in \cite{pan2014constant} is realized by unequal per-antenna power allocation to facilitate efficient precoding. Multiuser precoding \cite{mohammed2013per,zhang2016per} exploit spatial DoFs and enable efficient per-antenna envelope transmission with nonlinear RF components. Besides, PAPR-aware MU precoding in massive MIMO-OFDM systems jointly minimizes the MUI and peak power of the signal \cite{studer2013aware,cha2014generalized,bao2016efficient}. However, the minimization of peak power does not necessarily minimize the PAPR. In this paper, we take the PAPR of the transmit signal into account for optimization while keeping the MUI below a predetermined threshold.

Robust optimization of wireless communication systems has been extensively studied with imperfect channel knowledge \cite{vucic2009robust,tajer2011robust,shen2012distributed,zheng2009robust,wang2011robust,gharavol2010robust,rong2006robust,chalise2007robust,wang2011probabilistic,chung2011probabilistic}. It is not realistic to assume perfect channel state information (CSI) at the BS especially for densely-deployed, highly-mobilized 5G access channels. The imperfect CSI may be caused by inaccurate channel estimation over time and frequency, quantization errors, or offsets between reciprocal channels in time or frequency. Imperfect CSI and its impact on massive MIMO performance can be modeled as either bounded or stochastic errors. Bounded-based error robust MIMO beamforming was studied for broadcasting channels \cite{vucic2009robust}, multi-cell systems \cite{tajer2011robust,shen2012distributed},
and cognitive radio systems \cite{zheng2009robust,wang2011robust,gharavol2010robust} where the CSI errors are bounded. As a less conservative approach, stochastic robust beamforming was studied in \cite{rong2006robust,chalise2007robust,wang2011probabilistic,chung2011probabilistic} which assume the CSI errors are normally distributed. We consider imperfect CSI in both the bounded and stochastic sense in our robust PAPR-aware precoding framework and develop solutions under these uncertainties.

Instead of considering independent processing in each cell, multicell processing based on cooperation between BSs has emerged as a promising solution for suppressing co-channel interference \cite{huang2011distributed}. The optimal intercell coordination requires coherence between the signals from different BSs and the transmissions are controlled in a centralized manner. It has been shown that the coordination between neighboring BSs can improve capacity \cite{sawahashi2010coordinated,choi2008capacity,somekh2009cooperative}. This paper explores various formulations of the PAPR-aware precoding problem in a multicell context including the important cases of intercell coordination. 

The outline and contributions of this paper are as follows:
\begin{enumerate}
\item Basic PAPR-aware precoding optimization: Section III formulates an optimization framework to minimize instantaneous transmit power with the assumption of perfect CSI to achieve a predefined PAPR and MUI at the transmitter and receiver, respectively. To circumvent the nonconvexity, we \textit{approximate the solution using SDR and apply rank reduction to derive the rank-1 optimal solution}.
\item Robust PAPR-aware precoding optimization: Section IV incorporates both bounded and statistical CSI errors into a more realistic PAPR-aware precoding design by considering robust optimization techniques for \textit{coarse robust precoding, fine robust precoding via S-procedure and fine robust precoding via Bernstein-type inequality}. We show that the PAPR-aware robust precoding problems that incorporate the channel uncertainties can be formulated as SDR problems and efficiently solved.
\item PAPR-aware precoding for intercell coordination: Section V demonstrates the importance of PAPR-aware precoding to serve cell-edge users in three typical scenarios of intercell coordination: \textit{coherent transmission, fast cell selection, and interference coordination}. Using SDR, we achieve important reductions in cell-edge interference. We compare the computational complexities between the baseline and a variety of proposed approaches to illustrate the efficiency of our contribution.
\end{enumerate}

Section II provides the system model and Section VI and VII the simulation results and conclusions.
 
\textbf{Notation}: We denote vectors by boldface lowercase letters, e.g. $\mathbf{y}$, and matrices by boldface uppercase letters, e.g. $\mathbf{Y}$. The $i$th component of a vector $\mathbf{y}$ is $y_i$. Given two matrices $\mathbf{A}$ and $\mathbf{B}$, $\mathbf{A}\succ\mathbf{B}$ ($\mathbf{A}\succeq\mathbf{B}$) means that $\mathbf{A}-\mathbf{B}$ is positive definite (semidefinite). $\mathbb{C}$, $\mathbb{R}$ and $\mathbb{H}$ denote the complex, real and Hermitian sets, $\{\cdot\}^H$ the Hermitian transpose, $\textbf{E}(\cdot)$ the expectation, $\mathbf{0}_M$ the $M\times M$ zero matrix, $\mathbf{I}_M$ the $M\times M$ identity matrix, $\norm{x}_2^2$ the Euclidean norm, $\norm{x}_\infty$ the infinity norm, $\text{Tr}\{\cdot\}$ the trace of a matrix, $\Re\{\cdot\}$ the real part of a matrix, $\Im\{\cdot\}$ the image part of a matrix, $\text{diag}\{\cdot\}$ the block diagonal matrix whose diagonal elements are matrices, and $\otimes$ the Kronecker product. 
\begin{figure*}[!b]
\normalsize
\hrulefill
\begin{equation}
\mathbf{s}=\underbrace{\begin{bmatrix}
    \mathbf{H}_1 & \mathbf{0} & \dots  & \mathbf{0} \\
    \mathbf{0} & \mathbf{H}_2   & \dots  & \mathbf{0} \\
    \vdots & \vdots & \ddots & \vdots \\
    \mathbf{0} & \mathbf{0} & \dots  & \mathbf{H}_{N_c}
\end{bmatrix}}_{\triangleq\mathbf{H}}\underbrace{\begin{bmatrix}
    \mathbf{I}_{1,1} & \mathbf{I}_{1,2} & \dots  & \mathbf{I}_{1,N_t} \\
    \mathbf{I}_{2,1} & \mathbf{I}_{2,2}   & \dots  & \mathbf{I}_{2,N_t} \\
    \vdots & \vdots & \ddots & \vdots \\
    \mathbf{I}_{N_c,1} & \mathbf{I}_{N_c,2} & \dots  & \mathbf{I}_{N_c,N_t}
\end{bmatrix}}_{\triangleq\bar{\mathbf{P}}}
\underbrace{\begin{bmatrix}
    \mathbf{Q} & \mathbf{0} & \dots  & \mathbf{0} \\
    \mathbf{0} & \mathbf{Q}   & \dots  & \mathbf{0} \\
    \vdots & \vdots & \ddots & \vdots \\
    \mathbf{0} & \mathbf{0} & \dots  & \mathbf{Q}
\end{bmatrix}}_{\triangleq\bar{\mathbf{Q}}}
\underbrace{\left(
\begin{array}{c}
{\mathbf{x}}_1\\
{\mathbf{x}}_2\\
\vdots\\
{\mathbf{x}}_{N_t}
\end{array}
\right)}_{{\mathbf{x}}}\label{eq_shpqx},
\end{equation}
\vspace*{4pt}
\end{figure*}
\section{System Model and Problem Formulation}
\label{sec:sysmodel}
The purpose of the PAPR-aware massive MIMO-OFDM downlink precoding is to output a transmit signal that meets the predetermined PAPR and MUI targets. In order to quantify MUI, the overall massive MIMO-OFDM downlink precoding constraint needs to be formulated. Consider the downlink massive MIMO system of Fig. \ref{fig_block_sdr} which has $M_r$ single antenna users and one BS equipped with $N_t$ antennas. The number of BS antennas is significantly larger than the number of simultaneous users, $N_t\gg M_r$. 
\begin{figure}[h!]
\centering
\includegraphics[width=3.5in]{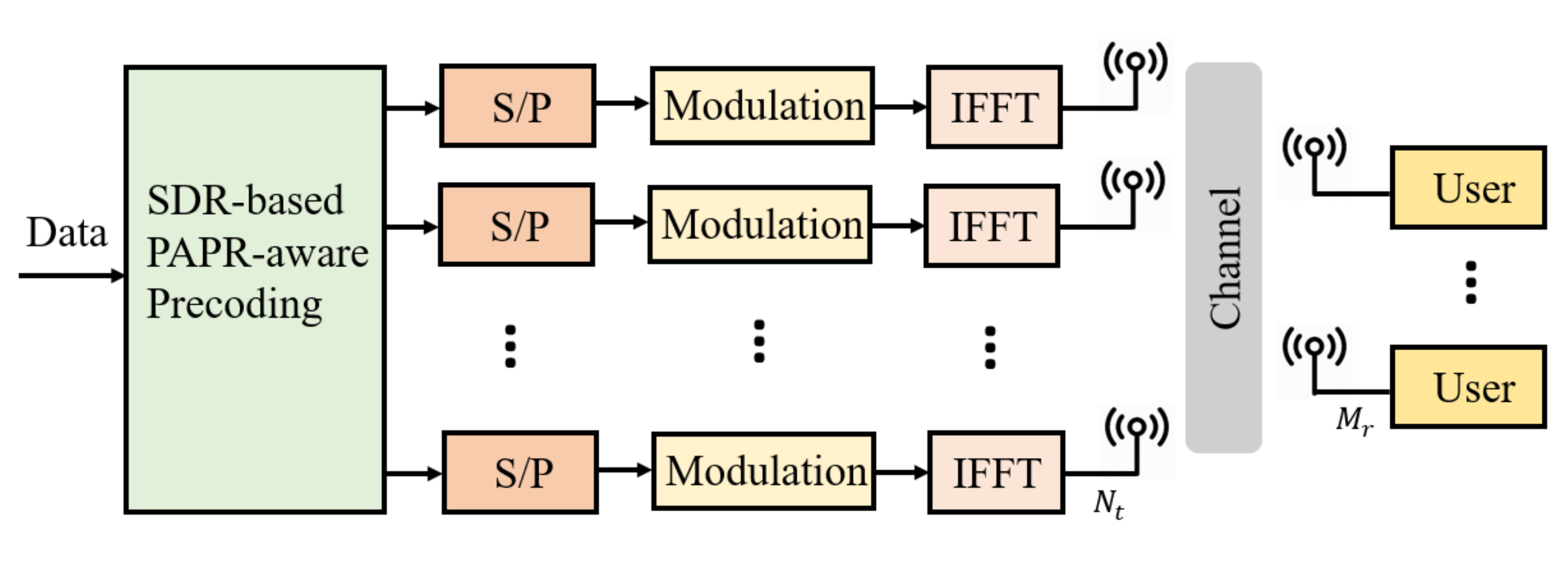}
\caption{System diagram of the proposed PAPR-aware downlink massive MIMO-OFDM system based on SDR.}
\label{fig_block_sdr}
\end{figure}
To illustrate the excess DoFs, enabled by the large scale transmit antenna array, we start with a single-carrier multiuser MIMO system and then extend it to the case of $N_c$ subcarriers. The single-carrier received signal $\dot{\mathbf{y}}\in\mathbb{C}^{M_r\times 1}$ can be represented as
\begin{equation}\nonumber
\dot{\mathbf{y}}=\dot{\mathbf{s}}+\left(\dot{\mathbf{H}}\dot{\mathbf{X}}-\dot{\mathbf{s}}\right)+\dot{\mathbf{w}},
\end{equation}
where $\dot{\mathbf{s}}\in\mathbb{C}^{M_r\times 1}$ represents the complex constellation before precoding, $\dot{\mathbf{H}}\in\mathbb{C}^{M_r\times N_t}$ the flat fading channel coefficients where the $(m,n)$ element represents the complex Gaussian channel tap between the $m$th user and $n$th BS antenna, $\dot{\mathbf{X}}\in\mathbb{C}^{N_t\times 1}$ the single-carrier transmit signals at the antennas, $\dot{\mathbf{w}}\in\mathbb{C}^{M_r\times 1}$ the additive white Gaussian noise, and $(\dot{\mathbf{H}}\dot{\mathbf{X}}-\dot{\mathbf{s}})$ the MUI. The entries of $\dot{\mathbf{H}}$ and $\dot{\mathbf{w}}$ are independent and identically distributed. The following precoding constraint must be satisfied to eliminate the MUI:
\begin{equation}\nonumber
\dot{\mathbf{s}}=\dot{\mathbf{H}}\dot{\mathbf{X}}.
\end{equation}

The channel matrix is underdetermined since we assume the number of BS antennas is significantly larger than the number of users. However, the formulation of an OFDM waveform is more complicated than the single-carrier case since PAPR is related to constructive and destructive addition of multiple carriers in the time domain, whereas precoding is related to the desired spatial characteristics. The overall problem of multicarrier precoding across different antennas and subcarriers is formulated in Section \ref{sec_precoding_form}.

\subsection{Multicarrier Frequency-Space Precoding Formulation}
\label{sec_precoding_form}
Assuming perfect channel state information is available at the transmitter (CSIT), linear precoding can be applied at the BS to eliminate the MUI at the receiver. Zero forcing (ZF) precoding has been widely applied in massive MIMO due to its simplicity and exceptional performance \cite{studer2013aware}. By applying ZF precoding for the $m$th subcarrier, the transmit symbols satisfy $\mathbf{s}_m=\mathbf{H}_m\mathbf{X}_m(1\leq m\leq N_c)$. In order to fully exploit the DoFs of the large scale array, the user information symbols at the BS need to be mapped to the transmit antennas so that the information symbols received by each user experience no interference from the other users \cite{studer2013aware,bao2016efficient}. 

\begin{figure}[h!]
\centering
\includegraphics[width=3.5in]{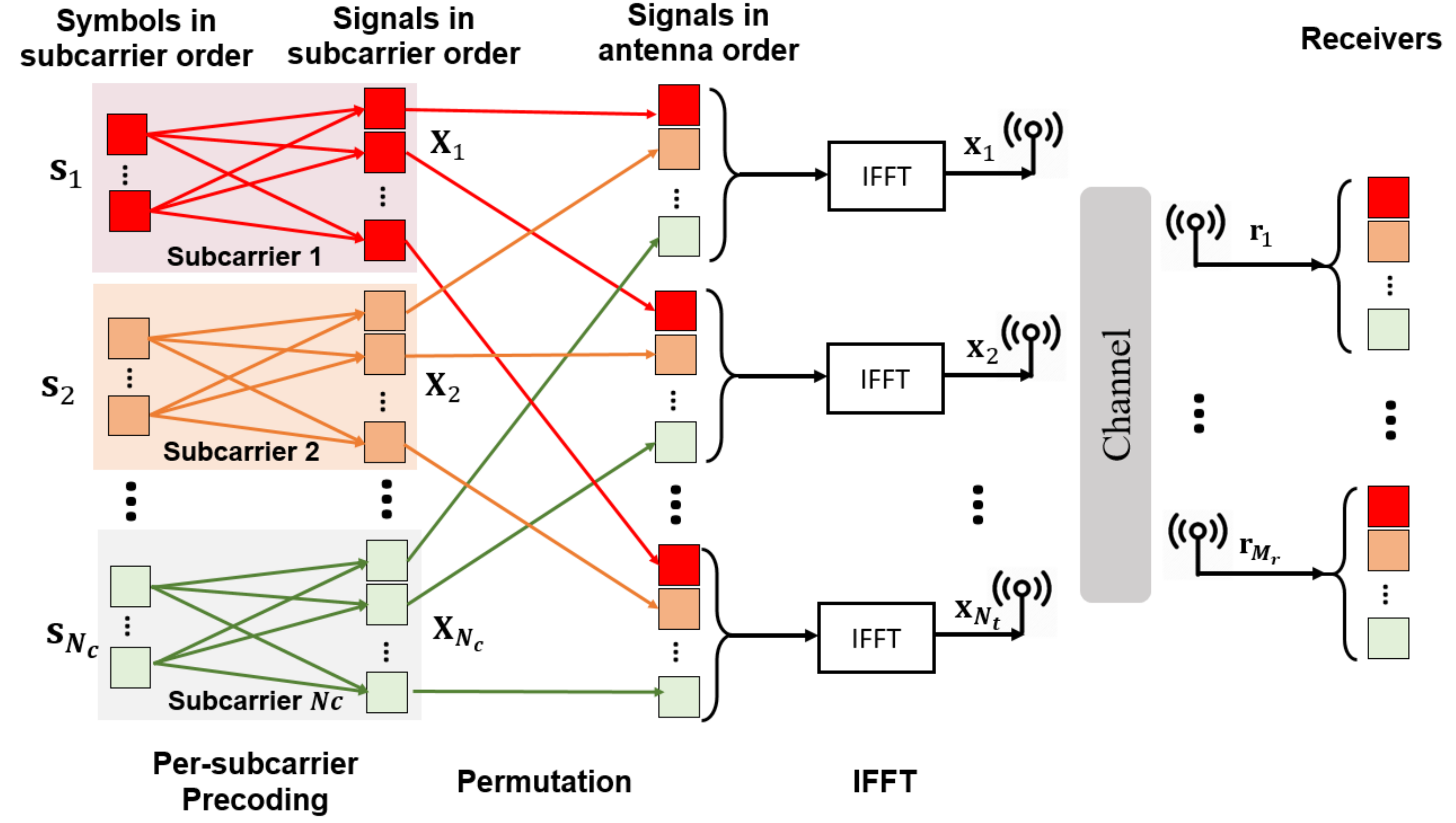}
\caption{Multicarrier frequency-space precoding formulation with permutation and IFFT.}
\label{fig_MU_precoding_form}
\end{figure}

As shown in Fig. \ref{fig_MU_precoding_form}, the frequency domain transmit symbols are permuted and $N_c$-point inverse fast Fourier transform (IFFT) operations are carried out for each antenna. As a result, the overall multicarrier frequency-space precoding can be represented as equation (\ref{eq_shpqx}) at the bottom of this page, where $\mathbf{s}$ represents the multicarrier information symbols, $\mathbf{H}_m$ the $M_r\times N_t$ channel matrix corresponding to the $m$th subcarrier, $\bar{\mathbf{P}}$ the permutation matrix, where $N_t\times N_c$ matrix $\mathbf{I}_{m,n}$ contains 0s expect for a 1 at element $(n,m)$ for $1\leq m \leq N_c$ and $1\leq n \leq N_t$, $\bar{\mathbf{Q}}$ the $N_x\times N_x$ diagonalization matrix comprised by FFT matrix $\mathbf{Q}$, where $N_x\triangleq N_tN_c$, and $\mathbf{x}$ the time-domain transmit signal. Note that the permutation matrix satisfies $\bar{\mathbf{P}}\bar{\mathbf{P}}^H=\mathbf{I}_{N_x}$ and $\bar{\mathbf{P}}^H$ can be considered as the reverse permutation matrix. Thus we have
\begin{equation}
\mathbf{s}={\mathbf{H}}\bar{\mathbf{P}}\bar{\mathbf{Q}}{\mathbf{x}}=\tilde{\mathbf{H}}\mathbf{x},
\label{eq_mui_remover}
\end{equation}
where $\tilde{\mathbf{H}}={\mathbf{H}}\bar{\mathbf{P}}\bar{\mathbf{Q}}$ represents the mapping operation of the time domain transmit signal to the information symbols of the overall massive MIMO system. 

\subsection{PAPR Basics}
In order to accommodate the large variations of the instantaneous transmit power at each antenna, the PA built for OFDM signals must have a wide linear region \cite{gomez2012link}. Since linear RF components are generally more expensive and less power efficient compared with their nonlinear counterparts, practical implementations of OFDM usually employ sophisticated PAPR reduction schemes. The power efficiency of PA is highly dependent on the PA architecture and input signal distribution, e.g. reference \cite{ochiai2012instantaneous} gives the power distribution of a single-carrier frequency-division multiple access (FDMA) signal. High input power back-off (IBO) is often required to keep the signal with high PAPR within the amplifier’s linear region and avoid in-band and out-of-band distortion \cite{yao2018dpd,vuk2018}. A highly linear (e.g. class A) PA with OFDM input and sufficient back-off to avoid out-of-band radiation could be even less than 5\%. Consider a complex baseband OFDM signal with $N_c$ subcarriers, the PAPR of signal ${\mathbf{x}}_m$ at the $m$th antenna is defined as the ratio between the maximum power and the average power of the entire OFDM signal \cite{ochiai2003performance}:
\begin{equation}
\text{PAPR}({\mathbf{x}}_m)=\frac{\norm{{\mathbf{x}}_m}_{\infty}^{2}}{P_{av}},
\label{eq3}
\end{equation}
where $P_{av}=E{|x(t)|^2}$ is defined as average transmitted power in statistical sense and $x(t)$ is the continuous transmitted signal. Note that the PAPR of complex baseband signal ${\mathbf{x}}_m$ satisfies $\text{PAPR}({\mathbf{x}}_m)\geq1$, and the equality is satisfied when all the symbols of $\mathbf{x}_m$ are constant as is the case with the PSK-type constellation \cite{ochiai2003performance} or constant envelope OFDM  \cite{mohammed2013per,zhang2016per,pan2014constant}. Since the distribution of PAPR would be more practical while the low-pass filters are applied in baseband signal processing unit of commercial wireless communication systems,  the true PAPR should be evaluated by oversampling \cite{ochiai2001distribution}. It is shown in \cite{tellambura2001computation} that the approximation error of PAPR is negligible when the oversampling factor is as large as 4. As a result, the oversampling operation is carried out in PAPR complementary cumulative distribution function (CCDF) comparisons in this paper to get better approximation of true PAPR of continuous and band-limited transmit signals.
\subsection{QCQP Problem Formulation With Perfect CSI}
The optimization problem of transmit power minimization subject to PAPR and MUI constraints is formulated as a QCQP problem in this subsection. We begin with the assumption of perfect CSI which will serve as a stepping stone towards the more realistic and robust PAPR-aware precoding scheme with imperfect CSI presented in the next section. The energy efficiency not only depends on the transmit power but also on the PAPR \cite{michailow2013low}. To derive an energy-efficient PAPR-aware precoding solution, we follow the design principles below:

\textbf{Principle 1}: The transmit power should be reduced as much as possible, but still provide the required signal-to-noise ratio (SNR) at the receiver. Lower transmit power reduces the power consumption of the BS and facilitates linear operation of the PA and hence lower out-of-band emissions.

\textbf{Principle 2}: Instead of minimizing the PAPR, set practical targets of PAPR for each PA since different PAPR ranges are allowed for different classes of PAs. This is also to avoid empty interior of the optimization problem.

\textbf{Principle 3}: Because of hardware impairments (e.g. amplifier nonlinearity, I/Q-imbalance, phase noise, quantization errors, etc), CSI uncertainties, and thermal noise \cite{bjornson2014massive}, a small amount of MUI will not significantly degrade the performance of massive MIMO-OFDM downlink precoding can be tolerated.

Based on these principles, our optimization strategy is to minimize the transmit power while keeping the PAPR of each antenna and MUI below the predetermined thresholds $\alpha_m(1\leq m\leq N_t)$ and $\delta_e$, respectively. Noticed that the symbol-wise (local) average power instead of statistically average power is applied in the PAPR representation, since the optimizations are carried out on symbol basis. Therefore, the problem can be formulated as:
\begin{subequations}
\begin{align}\label{p1_0}
(\textbf{P1}):\min_{\mathbf{x}} \;\;\;\;\;&\norm{\mathbf{x}}_2^2\\\label{p1_1}
\text{subject to} \;\;\;\;\;&\frac{\norm{\mathbf{F}_m{\mathbf{x}}}_{\infty}^{2}}{\norm{\mathbf{F}_m\mathbf{x}}_{2}^{2}}\leq \frac{\alpha_m}{N_c}, \qquad\forall m\in \{1,\cdots,N_t\}\\\label{p1_2}
&\norm{\mathbf{s}-\tilde{\mathbf{H}}\mathbf{x}}_{2}^{2}\leq\delta_e,
\end{align}
\end{subequations}
where $\mathbf{x}=[\mathbf{x}_1^H\;\cdots\;\mathbf{x}_{N_t}^H]^H$ and $\mathbf{x}_m=\mathbf{F}_m\mathbf{x}$.\footnote{$\mathbf{F}_m=[\underbrace{\mathbf{0}_{N_c},...,\mathbf{I}_{N_c},...,\mathbf{0}_{N_c}}_{N_t\;\text{submatrices}}]$, $\mathbf{I}_{N_c}$ is the $m$th submatrix of $\mathbf{F}_m$, the other submatrices of $\mathbf{F}_m$ are $\mathbf{0}_{N_c}$.} The MUI constraint of (\ref{p1_2}) is derived from (\ref{eq_mui_remover}) when a precoding error $\delta_e$ is allowed. The MIMO techniques with CSIT achieve multiplexing gain and hence DoFs. In addition, the deployments of massive MIMO systems which equipped BS with significantly larger number of antennas than users facilitate excessive DoFs when comparing with traditional MIMO systems. Particularly in the proposed optimization problem, the level of underdetermination of underdetermined matrix $\tilde{\mathbf{H}}$ is proportional to the number of BS antennas while assuming the number of user is known. The excessive DoFs enable us to select the transmit signals in  (\ref{p1_2}) from a larger solution space in the optimization \textbf{P1}. The proposal in  \cite{yao2018dpd} which exploits excessive DoFs to reduce the complexity of digital predistortion could explain how excessive DoFs are utilized in massive MIMO systems from a different perspective. Note that while the objective function for the power minimization in (\ref{p1_0}) and the MUI constraint in (\ref{p1_2}) are convex, the set of PAPR constraints in (\ref{p1_1}) are nonconvex. Although the constraints in (\ref{p1_1}) can be relaxed to convex form by neglecting the average power \cite{studer2013aware}\footnote{In \cite{studer2013aware}, constraint (\ref{p1_1}) is relaxed as $\norm{\mathbf{F}_m{\mathbf{x}}}_{\infty}^{2}\leq\beta_m,\forall m\in \{1,2,...,N_t\}$, and which is convex.}, this would result in a sub-optimal solution to the PAPR-aware precoding problem. As opposed to the sub-optimal convex relaxation in \cite{studer2013aware,bao2016efficient}, we convert (\ref{p1_1}) to quadratic form:
\begin{align}\label{eqn_papr5}\nonumber
\mathbf{x}^H\mathbf{F}_m^H\left(\mathbf{Y}_i-\frac{\alpha_m}{N_c}\mathbf{I}_{N_c}\right)\mathbf{F}_m\mathbf{x}\leq 0,&\forall i\in\{1,\cdots,N_c\},\\
 &\forall m\in\{1,\cdots,N_t\},
\end{align}
where $\mathbf{Y}_i\in\mathbb{R}^{N_c}$ is defined as 
\begin{equation}
\mathbf{Y}_i(i,j)=\left\{
                \begin{array}{ll}
                  1,j=i\\
                  0,j\neq i
                \end{array}
              \right..
\end{equation}
Equation (\ref{eqn_papr5}) can be written as
\begin{equation}
\mathbf{x}^H\mathbf{Q}_{im}\mathbf{x}\leq 0, \forall i,m\label{x_Q},
\end{equation}
by defining $\mathbf{Q}_{im}=\mathbf{F}_m^H\left(\mathbf{Y}_i-\frac{\alpha_m}{N_c}\mathbf{I}_{N_c}\right)\mathbf{F}_m$. Similarly, the constraints of (\ref{p1_2}) can be reformulated as
\begin{equation}
\begin{bmatrix}
    \mathbf{x}^H & t 
\end{bmatrix}
\begin{bmatrix}
    \tilde{\mathbf{H}}^H\tilde{\mathbf{H}}       & -\tilde{\mathbf{H}}^H\mathbf{s} \\
    -\mathbf{s}^H\tilde{\mathbf{H}}       & \mathbf{s}^H\mathbf{s} 
\end{bmatrix}
\begin{bmatrix}
   \mathbf{x}   \\
   t      
\end{bmatrix}\leq\delta_e,
\end{equation}
where $t=1$.

By defining $\tilde{\mathbf{x}}\triangleq\begin{bmatrix}
   \mathbf{x}   \\
   t      
\end{bmatrix}$ and ${\mathbf{G}}\triangleq\begin{bmatrix}
    \tilde{\mathbf{H}}^H\tilde{\mathbf{H}}       & -\tilde{\mathbf{H}}^H\mathbf{s} \\
    -\mathbf{s}^H\tilde{\mathbf{H}}       & \mathbf{s}^H\mathbf{s}     
\end{bmatrix}$, the constraints of (\ref{p1_2}) can be reduced to 
\begin{equation}
\tilde{\mathbf{x}}^H{\mathbf{G}}\tilde{\mathbf{x}}\leq\delta_e.\label{x_G}
\end{equation}
By combining (\ref{x_Q}) and (\ref{x_G}), the optimization problem \textbf{P1} becomes
\begin{subequations}
\begin{align}
(\textbf{P2}):\min_{\tilde{\mathbf{x}}} \;\;\;\;\;&\norm{\tilde{\mathbf{x}}}_2^2\\\label{p2_1}
\text{subject to} \;\;\;\;\;&\text{Tr}(\tilde{\mathbf{x}}^H\tilde{\mathbf{Q}}_{im}\tilde{\mathbf{x}})\leq 0,\forall i,m\\\label{p2_2}
&\text{Tr}(\tilde{\mathbf{x}}^H{\mathbf{G}}\tilde{\mathbf{x}})\leq\delta_e\\
&\tilde{x}_{N_x+1}=1,
\end{align}
\end{subequations}
where $\tilde{\mathbf{Q}}_{im}\triangleq
\left[
    \begin{array}{c;{2pt/2pt}c}
        \mathbf{Q}_{im} & \mathbf{0}_{N_x\times 1} \\ \hdashline[2pt/2pt]
        \mathbf{0}_{1\times N_x} & 0 
    \end{array}
\right]$ and $\tilde{x}_{N_x+1}=1$ indicates that the last entry of $\tilde{\mathbf{x}}$ is $1$. It is a nonconvex QCQP problem and can be relaxed to apply computationally efficient semidefinite programming (SDP) solutions. Note that the relationship between the objectives of \textbf{P1} and \textbf{P2} can be written as $\norm{\tilde{\mathbf{x}}}_2^2=\norm{\mathbf{x}}_2^2+1$ and the solutions of \textbf{P1} and \textbf{P2} as $\tilde{\mathbf{x}}_{opt}=\begin{bmatrix}
   \mathbf{x}_{opt}   \\
   1      
\end{bmatrix}$.

\section{Relaxed PAPR-Aware Massive MIMO Precoding using Semidefinite Programming}
In order to solve the previously derived nonconvex problems, one of the most common approaches is relaxing the nonconvex constraints to obtain a convex problem that approximates the original problem \cite{ma2004semidefinite}. This section shows that the derived QCQP for PAPR-aware massive MIMO precoding can be solved with the SDR method, before applying the randomization method for rank reduction.
\subsection{PAPR-Aware Massive MIMO Precoding Relaxation}
To derive the SDR of \textbf{P2}, a necessary step is to apply 
$\norm{\tilde{\mathbf{x}}}_2^2=\text{Tr}(\tilde{\mathbf{x}}\tilde{\mathbf{x}}^H)$, $\text{Tr}(\tilde{\mathbf{x}}^H\tilde{\mathbf{Q}}_{im}\tilde{\mathbf{x}})=\text{Tr}(\tilde{\mathbf{Q}}_{im}\tilde{\mathbf{x}}\tilde{\mathbf{x}}^H)$ and $\text{Tr}(\tilde{\mathbf{x}}^H{\mathbf{G}}\tilde{\mathbf{x}})=\text{Tr}({\mathbf{G}}\tilde{\mathbf{x}}\tilde{\mathbf{x}}^H)$. In particular, all the objective and constraints in \textbf{P2} are linear to the matrix $\bar{\mathbf{x}}\tilde{\mathbf{x}}^H$ after applying the SDR. As a result, we define a new variable as $\tilde{\mathbf{X}}\triangleq\tilde{\mathbf{x}}\tilde{\mathbf{x}}^H$. Therefore, $\tilde{\mathbf{X}}$ is a rank-1 symmetric positive semidefinite (PSD) matrix represented by $\text{rank}(\tilde{\mathbf{X}})=1$ and $\tilde{\mathbf{X}}\succeq 0$. The optimization problem \textbf{P2} is equivalent to 
\begin{subequations}
\begin{align}
(\textbf{P3}):\min_{\tilde{\mathbf{X}}} \;\;\;\;\;&\text{Tr}(\tilde{\mathbf{X}})\\\label{p3_1}
\text{subject to} \;\;\;\;\;&\text{Tr}(\tilde{\mathbf{Q}}_{im}\tilde{\mathbf{X}})\leq 0,\forall i,m\\\label{p3_2}
&\text{Tr}(\mathbf{G}\tilde{\mathbf{X}})\leq\delta_e\\\label{p3_3}
&\text{Tr}(\mathbf{O}_{x}\tilde{\mathbf{X}})=1\\\label{p3_4}
&\tilde{\mathbf{X}}\succeq 0\\\label{p3_5}
&\text{rank}(\tilde{\mathbf{X}})=1,
\end{align}
\end{subequations}
where $\mathbf{O}_x=\mathbf{o}_x\mathbf{o}_x^H$, $\mathbf{o}_x=[0,\cdots,0,1]^H\in\mathbb{R}^{(N_x+1)\times 1}$, and therefore $\text{Tr}(\mathbf{O}_{x}\tilde{\mathbf{X}})=\text{Tr}(\mathbf{o}_x^H\tilde{\mathbf{X}}\mathbf{o}_x)$. Notice that the objective function in \textbf{P3} minimizes the overall transmit power of the massive MIMO system and can be easily converted to a per-antenna power minimization problem when considering the fairness between transmit antennas \cite{yu2007transmitter}. 
However, the reformulation of \textbf{P3} is just as difficult to solve as problem \textbf{P2} since the rank constraint (\ref{p3_5}) is nonconvex although the objective function and all other constraints are convex in $\tilde{\mathbf{X}}$. By dropping the rank constraint we can obtain the following SDR version of \textbf{P2}:
\begin{align*}
(\textbf{P4}): \min_{\tilde{\mathbf{X}}} \;\;\;\;\;&\text{Tr}(\tilde{\mathbf{X}})\\
\text{subject to} \;\;\;\;\;&(\ref{p3_1})-(\ref{p3_4}).
\end{align*}
The SDR problem in \textbf{P4} is convex and thus does not suffer from local minima.  The optimization problem is relaxed and the optimal objective value of \textbf{P4} is always less or equal to the value of \textbf{P3},
\begin{equation}
\text{Tr}(\tilde{\mathbf{X}}_{opt})\leq\text{Tr}(\tilde{\mathbf{X}}_{opt}')\label{relaxopt},
\end{equation}
where $\tilde{\mathbf{X}}_{opt}$ is the optimal result of \textbf{P4} and $\tilde{\mathbf{X}}_{opt}'$ the optimal result of \textbf{P3}. The equality holds when there exists a rank 1 optimal solution $\tilde{\mathbf{X}}_{opt}=\tilde{\mathbf{x}}_{opt}\tilde{\mathbf{x}}_{opt}^H$. 
There is a corresponding convex Lagrange dual problem for any given problem, which yields a bound on the optimal value of the primal problem. The bound is tightest and the strong duality holds when the Slater condition is satisfied \cite{boyd2004convex}. The robust extensions of \textbf{P4}, which accounts for imperfect channel knowledge based on bounded-error and statistical models, are discussed in Section \ref{sec_robust_precoding}.
 
\subsection{Duality of the SDR}
\label{section3_b}
An SDR problem can be solved by employing the primal-dual path following algorithm, which has polynomial
complexity \cite{luo2010semidefinite}. Consider the Lagrange dual problem of \textbf{P2}, which can be written as
\begin{align}\nonumber
L&(\bm{\lambda},\nu_1,\nu_2,\mathbf{x})\\\nonumber
=&\tilde{\mathbf{x}}^H\tilde{\mathbf{x}}+\sum_{i=1}^{N_c}\sum_{m=1}^{N_t}\lambda_{im}\tilde{\mathbf{x}}^H\tilde{\mathbf{Q}}_{im}\tilde{\mathbf{x}}+\nu_1(\tilde{\mathbf{x}}^H\mathbf{G}\tilde{\mathbf{x}}-\delta_e)\\\nonumber
&-\nu_2(\tilde{\mathbf{x}}^H\mathbf{O}_{x}\tilde{\mathbf{x}}-1)\\\nonumber
=&\tilde{\mathbf{x}}^H\left(\mathbf{I}_{N_x+1}+\sum_{i=1}^{N_c}\sum_{m=1}^{N_t}\lambda_{im}\tilde{\mathbf{Q}}_{im}+\nu_1\mathbf{G}-\nu_2\mathbf{O}_{x}\right)\tilde{\mathbf{x}}\\
&+\nu_2-\nu_1\delta_e,
\end{align}
where $\bm{\lambda},\nu_1,\nu_2$ are non-negative Lagrangian multipliers. The Hessian matrix of the Lagrangian of the problem is derived as
\begin{equation}
\nabla_{\tilde{\mathbf{x}\mathbf{x}}^H}=\mathbf{I}_{N_x+1}+\sum_{i=1}^{N_c}\sum_{m=1}^{N_t}\lambda_{im}\tilde{\mathbf{Q}}_{im}+\nu_1\mathbf{G}-\nu_2\mathbf{O}_{x}.
\end{equation}
Therefore, both the dual problems of \textbf{P2} and \textbf{P4} are the same and can be written as
\begin{subequations}
\begin{align}
(\textbf{P5}):&\max_{\mathbf{V},\bm{\lambda},\nu_1,\nu_2} \;\;\;\;\;\nu_2-\nu_1\delta_e\\\label{p4_1}
\text{subject to} &\;\;\;\;\;\underbrace{\mathbf{I}_{N_x+1}+\sum_{i=1}^{N_c}\sum_{m=1}^{N_t}\lambda_{im}\tilde{\mathbf{Q}}_{im}+\nu_1\mathbf{G}-\nu_2\mathbf{O}_{x}}_{\triangleq\mathbf{V}}\succeq \mathbf{0}.
\end{align}
\end{subequations}
The Slater condition, which assumes the convex optimization problem has a nonempty interior, is usually supposed to be satisfied in semidefinite relaxation research problems in wireless communications \cite{ma2004semidefinite}. Specifically, in this work, we are looking for signals with minimum power in the interior set that satisfy both the precoding and practical PAPR constraints. As a result, we assume the condition of nonempty interior is satisfied and the strong duality holds for \textbf{P4}. As a result, we have
\begin{equation}
\text{Tr}(\tilde{\mathbf{X}}_{opt})=\nu_2^*-\nu_1^*\delta_e\label{strongdual},
\end{equation}
where $\tilde{\mathbf{X}}_{opt}$ is the optimal solution of \textbf{P4}, and $(\nu_1^*,\nu_2^*)$ is the optimal solution of \textbf{P5}. It is shown in (\ref{strongdual}) that given a larger allowable precoding error $\delta_e$, a lower optimal transmit power $\text{Tr}(\tilde{\mathbf{X}}_{opt})$ can be achieved. Intuitively, lower allowable precoding error means smaller feasible optimization solution region. Such a solution may increase the transmit power, but achieves less MUI at the receivers.
Upon defining $\bar{\mathbf{X}}\triangleq\begin{bmatrix}
    \Re\{\tilde{\mathbf{X}}\}       & -\Im\{\tilde{\mathbf{X}}\} \\
    -\Im\{\tilde{\mathbf{X}}\}       & \Re\{\tilde{\mathbf{X}}\}     
\end{bmatrix}$, $\bar{\mathbf{Q}}_{im}\triangleq\begin{bmatrix}
    \Re\{\tilde{\mathbf{Q}}_{im}\}       & -\Im\{\tilde{\mathbf{Q}}_{im}\} \\
    -\Im\{\tilde{\mathbf{Q}}_{im}\}       & \Re\{\tilde{\mathbf{Q}}_{im}\}     
\end{bmatrix}$, $\bar{\mathbf{G}}\triangleq\begin{bmatrix}
    \Re\{\mathbf{G}\}       & -\Im\{\mathbf{G}\} \\
    -\Im\{\mathbf{G}\}       & \Re\{\mathbf{G}\}     
\end{bmatrix}$, and $\bar{\mathbf{O}}_{x}\triangleq\begin{bmatrix}
    \mathbf{O}_{x}       & \mathbf{0} \\
    \mathbf{0}       & \mathbf{O}_{x}     
\end{bmatrix}$, the SDR problem \textbf{P4} can be transformed to a real-valued problem that can be addressed by the interior-point method. The details of the interior-point method can be found in \cite{ma2004semidefinite}.
\subsection{Approximation Error of the SDR}
In the previous analysis, we relaxed the PAPR-aware massive MIMO problem by dropping the nonconvex rank-1 constraint to formulate the optimization problem \textbf{P4}, and derive the optimal value $\tilde{\mathbf{X}}_{opt}$. The performance loss in terms of approximation error is evaluated in this subsection. If the rank of the optimal solution of \textbf{P4} is 1, there is no approximation error since we can further decompose it as $\tilde{\mathbf{X}}_{opt}=\tilde{\mathbf{x}}_{opt}\tilde{\mathbf{x}}_{opt}^H$, where $\tilde{\mathbf{x}}_{opt}$ is the optimal solution of \textbf{P2}. Without loss of generality, and assuming that $\text{rank}(\tilde{\mathbf{X}}_{opt})>1$, $\tilde{\mathbf{X}}_{opt}$ can be factorized using the Cholesky decomposition
\begin{equation}
\tilde{\mathbf{X}}_{opt}=\begin{bmatrix}
    \mathbf{U^H} \\
    \mathbf{u^H}    
\end{bmatrix}\begin{bmatrix}
    \mathbf{U}       & \mathbf{u}
\end{bmatrix}=
\begin{bmatrix}
    \mathbf{U^H}\mathbf{U}      & \mathbf{U^H}\mathbf{u} \\
    \mathbf{u^H}\mathbf{U}      &  \mathbf{u^H}\mathbf{u}     
\end{bmatrix}\label{eq_Uu},
\end{equation}
since being a Hermitian PSD matrix of rank $N_x$ or lower $\mathbf{U}\in\mathbb{C}^{N_x\times N_x}$, $\mathbf{u}\in\mathbb{C}^{N_x\times 1}$, and $\norm{\mathbf{u}}_2^2=1$. For now, suppose $\mathbf{x}_{opt}$ is the optimum of \textbf{P1}, we can derive that $\tilde{\mathbf{X}}_{opt}'$ is the optimum of problem \textbf{P3}, where
\begin{equation}
\tilde{\mathbf{X}}_{opt}'=
\left[
    \begin{array}{c;{2pt/2pt}c}
        \mathbf{x}_{opt}\mathbf{x}_{opt}^H & \mathbf{x}_{opt} \\ \hdashline[2pt/2pt]
        \mathbf{x}_{opt}^H & 1 
    \end{array}
\right].
\label{eq_xoptxopt}
\end{equation}
By comparing (\ref{eq_Uu}) and (\ref{eq_xoptxopt}), we have the following lemma:
\begin{lemma}\label{lemma_1}
Suppose $\Delta\mathbf{x}_u$ is the direct approximation error of $\mathbf{x}_{opt}$ ($\Delta\mathbf{x}_u=\mathbf{x}_{opt}-\mathbf{U^H}\mathbf{u}$). Then the upper bound of the approximation error satisfies $\norm{\Delta\mathbf{x}_u}_2^2\leq2\delta_e/|\lambda_{min}({\mathbf{H}}{\mathbf{H}}^H)|$, where $\lambda_{min}(\cdot)$ represents the minimum eigenvalue of a matrix.
\end{lemma}
\begin{proof}
See Appendix A.
\end{proof}

Remarks:
\begin{itemize}
\item The upper bound of the approximation error is determined by both the maximum allowed precoding error (a.k.a. MUI allowance) $\delta_e$ and the minimum eigenvalue of ${\mathbf{H}}{\mathbf{H}}^H$, which are independent ($\mathbf{H}$ is defined in (\ref{eq_shpqx}), the distribution of $\lambda_{min}({\mathbf{H}}{\mathbf{H}}^H)$ can be found in \textit{Proposition 3.6} of \cite{man2010probabilistic});
\item The approximation error of the SDR will be sufficiently low for a sufficiently low $\delta_e$, which is a preset value in the aforementioned optimization problems \textbf{P1}-\textbf{P4};
\item Due to the duality of the optimization, a lower optimal transmit power is achieved with higher $\delta_e$.
\end{itemize}
\subsection{Rank-1 Solution via Randomization}
We are able to use the SDR solution of \textbf{P4} to approximate the solution of the original problem \textbf{P1} directly with error $\Delta\mathbf{x}_u$. However, it is still necessary to find a more precise method by deriving the solution of \textbf{P1} through the solution of \textbf{P4}. This is especially true when the maximum allowed precoding error $\delta_e$ is large. Randomization is widely applied to extract an approximate QCQP solution from an SDR solution $\tilde{\mathbf{X}}_{opt}$\cite{ma2002quasi,ma2004semidefinite,goemans1995improved,sidiropoulos2006semidefinite,so2008unified}. The randomization generates a set of candidate vectors $\{\tilde{\mathbf{x}}_i\}$ using $\tilde{\mathbf{X}}_{opt}$ and chooses the best solution from these candidate vectors. The overall SDR and rank reduction algorithm based on randomization is described below.
\begin{algorithm}[H]
 \caption{SDR and rank reduction algorithm via randomization\cite{so2008unified}}
  \begin{algorithmic}[1]
 \renewcommand{\algorithmicrequire}{\textbf{Input:}}
 \renewcommand{\algorithmicensure}{\textbf{Output:}} 
 \REQUIRE  $\tilde{\mathbf{Q}}_{im}$, $\mathbf{G}$, $\mathbf{O}_{x}$, $M$ (number of randomization iterations)
 \ENSURE  $\tilde{\mathbf{x}}_{opt}$
 \\ Solve the SDR problem (\textbf{P4}), and obtain its solution $\tilde{\mathbf{X}}_{opt}$;
 \\ Factorize $\tilde{\mathbf{X}}_{opt}=\mathbf{V}^H\mathbf{V}$;
 \FOR{$i=1,2,...,M$}
    \STATE{Generate $\bm{\xi}_i\sim\mathcal{N}(0,1)$};
	\STATE{Compute $\tilde{\mathbf{x}}_{i}=\mathbf{V}^H\bm{\xi}_i$};
  \ENDFOR
 \\Choose $k=\text{arg}\min_{i=1,...,M}(\norm{\tilde{\mathbf{x}}_{i}}_2^2)$;
 \RETURN $\tilde{\mathbf{x}}_{opt}=\tilde{\mathbf{x}}_{k}$.
 \end{algorithmic} \label{algo_random}
 \end{algorithm}
The mathematical proof and approximation quality of Algorithm \ref{algo_random} can be found in \cite{so2008unified}, which shows that the optimal rank-1 solution exists for a large enough $M$.
Note that the complexity of the randomization process is much smaller relative to the SDR approach \cite{so2008unified}, and hence its computational complexity can be ignored in the complexity comparison in Section \ref{SDR_sim}.
\section{Robust PAPR-Aware Precoding With Imperfect CSI}
\label{sec_robust_precoding}
Robust optimization is usually developed to address either the bounded or statistical CSI errors. We propose three approaches to achieve robustness under channel uncertainties, namely coarse robust precoding (bounded-error), fine robust precoding via the S-procedure (bounded-error), and fine robust precoding via the Bernstein-type inequality (statistical error), respectively. In this paper, the transmit signals are optimized to meet the target PAPRs $\alpha_1,\cdots,\alpha_{N_t}$  and maximum allowed precoding error $\delta_e$ for every possible CSI error, and so facilitate robust precoding. In particular, robust precoding can be achieved by operating under the bounded-error channel conditions if the CSI errors are bounded. Alternatively, robust precoding can also be ensured in the probabilistic sense if the channel uncertainties are modeled as statistical distributions. The channel matrix is represented as below to model the channel estimation inaccuracies in massive MIMO systems
\begin{equation}
\tilde{\mathbf{H}}=\hat{\mathbf{H}}+\Delta\tilde{\mathbf{H}},
\end{equation}
where $\hat{\mathbf{H}}$ captures the channel matrix which is measured at the BS from the uplink channel, and $\Delta\tilde{\mathbf{H}}$ the channel uncertainty. 
Note that from (\ref{eq_mui_remover}), we get
\begin{align}\nonumber
\textbf{E}\left(\Delta\tilde{\mathbf{H}}\Delta\tilde{\mathbf{H}}^H\right)&=\textbf{E}\left(\Delta{\mathbf{H}}\bar{\mathbf{P}}\bar{\mathbf{Q}}\bar{\mathbf{Q}}^H\bar{\mathbf{P}}^H\Delta{\mathbf{H}}^H\right)\\
&=\textbf{E}\left(\Delta{\mathbf{H}}\Delta{\mathbf{H}}^H\right),
\end{align}
where $\Delta{\mathbf{H}}=\text{diag}\{\Delta\mathbf{H}_1,\cdots,\Delta\mathbf{H}_{N_c}\}$. $\Delta\mathbf{H}_m$ indicates the $M_r\times N_t$ channel error matrix of the $m$th subcarrier. In this section, we consider both the deterministic and stochastic model of the channel uncertainty, and derive the respective robust precoding schemes. 
\begin{itemize}
\item{\textit{Deterministic model (bounded-error):} } Here we assume the error matrix $\Delta\tilde{\mathbf{H}}$ take values from the bounded set
\begin{equation}
\Upsilon_d:=\left\{\text{Tr}(\Delta\tilde{\mathbf{H}}\Delta\tilde{\mathbf{H}}^H)\leq\epsilon_h^2\right\}\label{upsilon_d},
\end{equation}
where $\epsilon_h>0$ denotes the radius of the feasible region $\Upsilon_d$, which is associated with the degree of uncertainty of equivalent channel measurement $\hat{\mathbf{H}}$.
\item{\textit{Stochastic model:}} Here we assume the channel uncertainty subject to zero-mean complex Gaussian distribution with variance vector $\mathbf{R}_{\epsilon}$, that is
\begin{equation}
\Upsilon_p:=\left\{\text{vec}(\Delta\tilde{\mathbf{H}})\sim \mathcal{CN}(\mathbf{0},\mathbf{R}_{\epsilon})\right\}\label{upsilon_p},
\end{equation}
where $\text{vec}(\cdot)$ represents the column-by-column matrix vectorization.
\end{itemize}
The shape of the feasible region depends on the second-order statistics of the channel uncertainty, and the specific channel estimation approach. Since the channel information is only involved in the constraint (\ref{p3_2}), we need to guarantee no violation of constraint (\ref{p3_2}) to ensure robust precoding. 
\subsection{Coarse Robust Precoding}
We first adopt the bounded-error model and develop the coarse robust precoding for the bounded-error channel uncertainty. Since the constraints need to be satisfied for all possible CSI errors, the upper bound on the lefthand side of (\ref{p3_2}) should be less than the maximum allowed precoding error $\delta_e$. As a result, by applying the Cauchy-Schwarz inequality, the constraint (\ref{p3_2}) under the circumstances of coarse robust precoding can be represented as
\begin{align}\nonumber
\norm{\mathbf{s}-\left(\hat{\mathbf{H}}+\Delta\tilde{\mathbf{H}}\right)\mathbf{x}}_2^{2}&\leq\norm{\mathbf{s}-\hat{\mathbf{H}}\mathbf{x}}_2^{2}+\norm{\Delta\tilde{\mathbf{H}}\mathbf{x}}_2^{2}\\
&\leq\text{Tr}(\hat{\mathbf{G}}\tilde{\mathbf{X}})+\epsilon_h^2\left(\text{Tr}(\tilde{\mathbf{X}})-1\right),
\end{align} 
where $\hat{\mathbf{G}}\triangleq\begin{bmatrix}
    \hat{\mathbf{H}}^H\hat{\mathbf{H}}       & -\hat{\mathbf{H}}^H\mathbf{s} \\
    -\mathbf{s}^H\hat{\mathbf{H}}       & \mathbf{s}^H\mathbf{s}     
\end{bmatrix}$. Therefore, the constraint in terms of coarse robust precoding becomes
\begin{equation}
\text{Tr}\left((\hat{\mathbf{G}}+\epsilon_h^2\mathbf{I}_{N_x+1})\tilde{\mathbf{X}}\right)\leq\delta_e+\epsilon_h^2,
\label{crp}
\end{equation}
which is a standard SDR constraint. Therefore, the optimization problem \textbf{P4} with coarse robust precoding can be reformulated by replacing (\ref{p3_2}) with (\ref{crp}):
\begin{align*}
(\textbf{P6}):\min_{\tilde{\mathbf{X}}} \;\;\;\;\;&\text{Tr}(\tilde{\mathbf{X}})\\
\text{subject to} \;\;\;\;\;&(\ref{p3_1}),(\ref{p3_3}),(\ref{p3_4}),(\ref{crp}).
\end{align*}
The problem \textbf{P6} provides a robust optimization with a loose upper bound of precoding error. The advantage of coarse robust precoding lies in the fact that the change between the constraints (\ref{p3_2}) and (\ref{crp}) is not significant.
\subsection{Fine Robust Precoding via S-Procedure}
The loose upper bound derived in (\ref{crp}) for coarse robust precoding degrades the performance of robust optimization because it changes the feasibility region. As a result, a fine robust precoding with a tighter upper bound is proposed here. The fine robust precoding is derived as
\begin{equation}
\norm{\mathbf{s}-\left(\hat{\mathbf{H}}+\Delta\tilde{\mathbf{H}}\right)\mathbf{x}}_2^{2}=\text{Tr}\left((\hat{\mathbf{H}}_s+\Delta \mathbf{H}_s)\mathbf{W}_x(\hat{\mathbf{H}}_s+\Delta \mathbf{H}_s)^H\right),
\label{con_s_procedure}
\end{equation}
where we define $\hat{\mathbf{H}}_s\triangleq\begin{bmatrix}
    \hat{\mathbf{H}}       & -\mathbf{s}    
\end{bmatrix}$, $\Delta \mathbf{H}_s=\begin{bmatrix}
    \Delta\tilde{\mathbf{H}}       & \mathbf{0}    
\end{bmatrix}$ and $\mathbf{W}_x=\tilde{\mathbf{X}}|_{t=1}$. It does not sacrifice the performance degradation of the optimization as does its coarse counterpart.
\begin{lemma}
Given the bounded channel error matrix in (\ref{upsilon_d}), the constraint in (\ref{con_s_procedure}) can be relaxed to
\begin{subequations}
\begin{align}\label{x-y}
&\delta_e-\lambda\epsilon_h^2-\text{Tr}\left((\tilde{\mathbf{X}}-\tilde{\mathbf{Z}})\hat{\mathbf{H}}_s^H\hat{\mathbf{H}}_s\right)\geq 0\\
&\tilde{\mathbf{X}}(-\tilde{\mathbf{X}}+\lambda\mathbf{I}_{N_x+1})^{-1} \tilde{\mathbf{X}}\succeq \tilde{\mathbf{Z}}\label{succ_y},
\end{align}
\end{subequations}
\label{lemma_deter}
\end{lemma}
where $\lambda>0$.
\begin{proof}
See Appendix B.
\end{proof}
By applying Schur's complement \cite{horn2012matrix}, (\ref{succ_y}) can be further represented as
\begin{equation}
\begin{bmatrix}
    \tilde{\mathbf{X}}-\lambda\mathbf{I}_{N_x+1} &  \tilde{\mathbf{X}}    \\
    \tilde{\mathbf{X}}       &  -\tilde{\mathbf{Z}}
\end{bmatrix}
\succeq \mathbf{0}.\label{succ_xy}
\end{equation}
Therefore, the optimization problem P4 with fine robust precoding can be represented as
\begin{align*}
(\textbf{P7}):\min_{\tilde{\mathbf{X}},\mathbf{Y},\lambda>0} \;\;\;\;\;&\text{Tr}(\tilde{\mathbf{X}})\\
\text{subject to} \;\;\;\;\;&(\ref{p3_1}),(\ref{p3_3}),(\ref{p3_4}),(\ref{x-y}),(\ref{succ_xy}).
\end{align*}
\subsection{Fine Robust Precoding using the Bernstein-Type Inequality}
With the assumption of statistical channel uncertainty model, we propose a less conservative reformulation with tractable probabilistic constraint of fine robust precoding using the Bernstein-type inequality \cite{bechar2009bernstein}. The constraints of (\ref{sdphwh}) with the stochastic model can be expressed as
\begin{align}\label{bti}
\text{Pr}\Bigg(\text{Tr}\left(\hat{\mathbf{H}}_s\mathbf{W}_x\hat{\mathbf{H}}_s^H\right)+\Delta\mathbf{h}_s^H(\mathbf{I}_{N_s}\otimes\mathbf{W}_x)\Delta\mathbf{h}_s&+\\\nonumber
2\text{Re}\Big(\text{vec}(\hat{\mathbf{H}}_s\mathbf{W}_x)^H\Delta\mathbf{h}_s\Big)\leq\delta_e\Bigg)&\geq 1-\gamma.
\end{align}
\begin{lemma}\label{lemma_statistical}
The constraint (\ref{bti}) is equivalent to a set of constraints below
\begin{subequations}
\begin{align}\label{bti_c_1}
&\text{Tr}\Big({\mathbf{R}_{\epsilon}'}(\mathbf{I}_{N_s}\otimes\tilde{\mathbf{X}})\Big)-\sqrt{-2\theta}\lambda_1-\theta \lambda_2+\delta_e\\\nonumber
&\quad\quad\quad\quad\quad\quad\quad\quad\quad\quad\quad\quad\quad-\text{Tr}\left(\hat{\mathbf{H}}_s\tilde{\mathbf{X}}\hat{\mathbf{H}}_s^H\right)\geq 0\\\label{bti_c_2}
&\sqrt{\norm{{\mathbf{R}_{\epsilon}^{'1/2}}(\mathbf{I}_{N_s}\otimes\tilde{\mathbf{X}})\mathbf{R}_{\epsilon}^{'1/2}}_2^2+2\norm{{\mathbf{R}_{\epsilon}^{'1/2}}\text{vec}(\hat{\mathbf{H}}_s\mathbf{W}_x)}_2^2}\quad\leq \lambda_1\\\label{bti_c_3}
&\lambda_2\mathbf{I}_{N_x N_s}-{\mathbf{R}_{\epsilon}^{'1/2}}(\mathbf{I}_{N_s}\otimes\tilde{\mathbf{X}})\mathbf{R}_{\epsilon}^{'1/2}\succeq\mathbf{0}\\\label{bti_c_4}
&\lambda_2\geq 0.
\end{align}
\end{subequations}
\end{lemma}
\begin{proof}
See Appendix C.
\end{proof}
The optimization problem with fine robust precoding in the statistical sense can then be represented as
\begin{align*}
(\textbf{P8}):\min_{\tilde{\mathbf{X}},\lambda_1,\lambda_2} \;\;\;\;\;&\text{Tr}(\tilde{\mathbf{X}})\\
\text{subject to} \;\;\;\;\;&(\ref{p3_1}),(\ref{p3_3}),(\ref{p3_4}),(\ref{bti_c_1})-(\ref{bti_c_4}).
\end{align*}
Note that (\ref{bti_c_1}), (\ref{bti_c_2}) and (\ref{bti_c_3}) are a linear constraint, a convex second-order cone (SOC) constraint, and a convex PSD constraint, respectively. This problem can be solved efficiently by using the interior-point method of \cite{boyd2004convex}.
\section{PAPR-Aware Precoding with Intercell Coordination}
\label{IC}
In this section, we incorporate the interference between the co-channel users into the proposed SDR framework and formulate three different transmission schemes for PAPR-aware precoding with intercell coordination: coherent transmission, fast cell selection and interference coordination. These schemes are illustrated in Fig. \ref{fig_intercoordination}. The single cell massive MIMO system that has been discussed in this paper so far effectively serves cell-center users. However, a practical limitation for enabling energy-efficient cellular network operation is the power consumption of the BS while serving cell-edge users. The users at the cell edges not only suffer from high path loss, but also from severe co-channel (intercell) interference when compared with the cell-center users \cite{sawahashi2010coordinated}\cite{irmer2011coordinated}. The PAPR-aware precoding with incercell coordination is based on the following observations:

\textbf{Observation 1}: PAPR-aware precoding is critical for serving cell-edge users. The BS allocates more power to improve SNR and satisfy the quality of service (QoS) of cell-edge users. PAPR reduction is especially necessary when the average power is high (high operating point), since expensive PAs are necessary to achieve high output power with linearity, and even then they may not be power efficient.

\textbf{Observation 2}: Downlink precoding is beneficial and feasible for intercell coordination when radio resource management is carried out in a centralized manner. Spatial DoFs available at the adjacent cells can be used to attenuate or mitigate the interferences (interference rejection) when the CSI of cell-edge users of adjacent cells are available at the centralized baseband processing unit (CBPU). Note that the interference is jointly suppressed across BSs by the CBPU rather than individually optimized at each BS.

The co-channel interference can be efficiently suppressed by the downlink beamforming optimization with available CSI at the BS. Centralized radio resource management allows for the optimization and scheduling to be done globally at the CBPU.  

\begin{figure}[h!]
\centering
\includegraphics[width=3.5in]{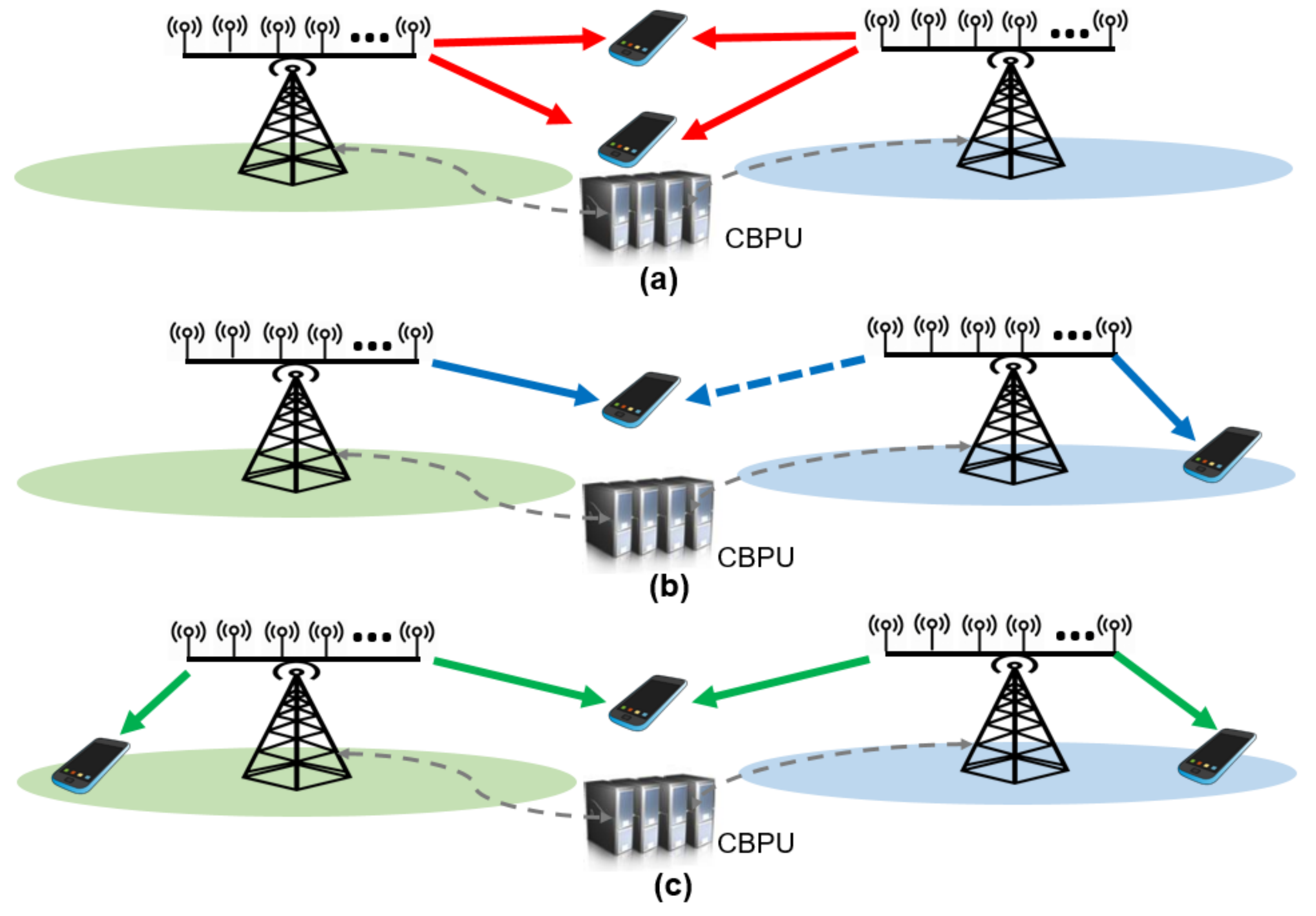}
\caption{Intercell coordination for three transmission schemes, (a) coherent transmission, (b) fast cell selection, and (c) interference coordination.}
\label{fig_intercoordination}
\end{figure}

\subsection{Coherent Massive MIMO Downlink Transmission} 
Coherent multipoint transmission is realized by a simultaneous transmission of signals to a user from multiple cooperating BSs and the contributions of the BSs are coherently combined at the user terminal. The advantages of coherent transmission are that the received SNR is increased and the spatial diversity enhanced. Spatial diversity that combines multiple independent fading paths reduces overall channel fading compared to a single link. The constraint (\ref{p3_2}) in \textbf{P4} under the circumstances of coherent transmission is then given as
\begin{align}
\norm{\mathbf{s}-\sum_{k=1}^{M_c}\tilde{\mathbf{H}}_k\mathbf{x}}_2^{2}\leq\delta_C,
\end{align} 
where $\delta_C$ denotes maximum precoding error allowance for coherent transmission and $\tilde{\mathbf{H}}_k$ the channel information between BS $k$ and all users. The equation can be reformulated as
\begin{align}
\text{Tr}(\mathbf{Z}_1\tilde{\mathbf{X}})\leq\delta_C,
\label{eq_ct}
\end{align}
where 
\begin{align}
\mathbf{Z}_1\triangleq\begin{bmatrix}
    (\sum_{k=1}^{M_c}\tilde{\mathbf{H}}_k)^H(\sum_{k=1}^{M_c}\tilde{\mathbf{H}}_k)       & -(\sum_{k=1}^{M_c}\tilde{\mathbf{H}}_k)^H\mathbf{s} \\
    -\mathbf{s}^H(\sum_{k=1}^{M_c}\tilde{\mathbf{H}}_k)       & \mathbf{s}^H\mathbf{s}     
\end{bmatrix}.
\end{align}
The total number of BSs involved in the interference coordination is $M_c$. The constraint (\ref{p3_2}) in \textbf{P4} is replaced by (\ref{eq_ct}), and the centralized radio resource management is applied to coordinate the optimization of a cluster of BSs.  
\subsection{Fast Massive MIMO Cell Selection}
The performances of cell-edge users are inevitably degraded by co-channel interference of adjacent cells because of the full frequency reuse. However, cell-edge users can be served dynamically by a selected BS through fast scheduling at the CBPU as shown in Fig. \ref{fig_intercoordination} (b). Therefore, one possible problem formulation is to minimize the interference level of cell-edge users by limiting the interference level from adjacent cells to some threshold. This fast cell selection technique is similar to nulling, where all the precoders are forced to have nulls placed toward other users. The constraint of a group of cell-edge users under the circumstances of fast cell selection can be written as
\begin{align}
\norm{\tilde{\mathbf{H}}_{e,m}\mathbf{x}}_2^{2}\leq\beta_s\lambda_{max}(\tilde{\mathbf{H}}_{e,m}^H\tilde{\mathbf{H}}_{e,m})\text{Tr}(\mathbf{x}\mathbf{x}^H),
\end{align} 
where $\tilde{\mathbf{H}}_{e,m}$ represents the channel between the $m$th BS and the group of cell-edge users. The expected total interference power received by the neighboring cell users can thus be limited to a fraction $\beta_s$ of the worst-case interference, where $\lambda_{max}(\cdot)$ represents the maximum eigenvalue of a matrix.
 Above constraint can be rewritten as
\begin{align}
\beta_s\lambda_{max}(\tilde{\mathbf{H}}_{e,m}^H\tilde{\mathbf{H}}_{e,m})\text{Tr}(\tilde{\mathbf{X}})-\text{Tr}(\mathbf{Z}_{e,m}\tilde{\mathbf{X}})\geq \beta_s\lambda_{max}(\tilde{\mathbf{H}}_{e,m}^H\tilde{\mathbf{H}}_{e,m}),
\label{eq_fast}
\end{align}
where 
\begin{align}
\mathbf{Z}_{e,m}\triangleq\begin{bmatrix}
    \tilde{\mathbf{H}}_{e,m}^H\tilde{\mathbf{H}}_{e,m}       & \mathbf{0} \\
    \mathbf{0}       & 0     
\end{bmatrix}.
\end{align}
Therefore, the constraint of (\ref{eq_fast}) is added to the optimization problem \textbf{P4} as an additional constraint in each BS.
\begin{algorithm}[H]
 \caption{Intercell Coordination for Different Scenarios with CBPU}
  \begin{algorithmic}[1] 
  \STATE{\textbf{Given} the scenario of transmission scheme is determined by CBPU, and the index of selected cell if fast cell selection is applied};
  \IF{$<$coherent transmission$>$}
   \STATE {Solve $\min_{\tilde{\mathbf{X}}}\text{Tr}(\tilde{\mathbf{X}})\;\;\;\;
\text{subject to} \;\;\;\;(\ref{p3_1}),(\ref{p3_3})-(\ref{p3_4}),(\ref{eq_ct})$};
	\STATE{Derive $\tilde{\mathbf{x}}_{opt}^{CT}$ with Algorithm \ref{algo_random}};
	\STATE{$\tilde{\mathbf{x}}_{opt}\leftarrow\tilde{\mathbf{x}}_{opt}^{CT}$};
  \ELSIF {$<$fast cell selection$>$}
   \STATE{$j\leftarrow$index of selected cell};
 	\FOR{$i=1,\cdots,j-1,j+1,\cdots,M_c$}
   		\STATE {Solve $\min_{\tilde{\mathbf{X}}_i}\text{Tr}(\tilde{\mathbf{X}}_i)\;\;\;\;\text{subject to} \;\;\;\;(\ref{p3_1})-(\ref{p3_4}),(\ref{eq_fast})$};
		\STATE{Derive $\tilde{\mathbf{x}}_{opt,i}^{FS}$ with Algorithm \ref{algo_random}};
  \ENDFOR
  \STATE {Solve $\min_{\tilde{\mathbf{X}}_j}\text{Tr}(\tilde{\mathbf{X}}_j)\;\;\;\;\text{subject to} \;\;\;\;(\ref{p3_1})-(\ref{p3_4})$};
  \STATE{Derive $\tilde{\mathbf{x}}_{opt,j}^{FS}$ with Algorithm \ref{algo_random}};
  \STATE{$\tilde{\mathbf{x}}_{opt}\leftarrow[(\tilde{\mathbf{x}}_{opt,1}^{FS})^H,\cdots,(\tilde{\mathbf{x}}_{opt,M_c}^{FS})^H]^H$};
  \ELSIF {$<$interference coordination$>$}
   \STATE {Solve $\min_{\tilde{\mathbf{X}}}\text{Tr}(\tilde{\mathbf{X}})\;\;\;\;
\text{subject to} \;\;\;\;(\ref{p3_1}),(\ref{p3_3})-(\ref{p3_4}),(\ref{eq_ic3})$};
	\STATE{Derive $\tilde{\mathbf{x}}_{opt}^{IC}$ with Algorithm \ref{algo_random}};
	\STATE{$\tilde{\mathbf{x}}_{opt}\leftarrow\tilde{\mathbf{x}}_{opt}^{IC}$};
  \ENDIF
 \RETURN $\tilde{\mathbf{x}}_{opt}$.
 \end{algorithmic} \label{algo_intercell}
 \end{algorithm}
\subsection{Massive MIMO Interference Coordination}
In contrast to coherent transmission and fast cell selection, interference suppression based coordinated precoding is a more general case and refers to a coordinated selection of the transmit precoders in each cell. This technique aims at eliminating or reducing the effect of intercell interference. The constraint (\ref{p3_2}) in \textbf{P4} under the circumstances of interference coordination precoding is given as
\begin{align}
\norm{\mathbf{s}-\sum_{k=1}^{M_c}\tilde{\mathbf{H}}_k\mathbf{x}_k}_2^{2}&\leq\delta_I,
\end{align} 
where $\tilde{\mathbf{H}}_k$ captures the channel information between BS $k$ and all users, and  $\mathbf{x}_k$ the transmit signal of BS $k$. It can be further reformulated as
\begin{align}
\text{Tr}(\mathbf{Z}_3\tilde{\mathbf{X}})\leq\delta_I,
\label{eq_ic3}
\end{align}
where $\delta_I$ denotes the maximum precoding error allowance for interference coordination, and
\begin{align}
\mathbf{Z}_3\triangleq\begin{bmatrix}
    \tilde{\mathbf{H}}_I^H\tilde{\mathbf{H}}_I       & -\tilde{\mathbf{H}}_I^H\mathbf{s} \\
    -\mathbf{s}^H\tilde{\mathbf{H}}_I       & \mathbf{s}^H\mathbf{s}     
\end{bmatrix}.
\end{align}
$\tilde{\mathbf{H}}_I\triangleq\text{diag}\{\tilde{\mathbf{H}}_1,\cdots,\tilde{\mathbf{H}}_{M_c}\}$, $\mathbf{x}\triangleq[\mathbf{x}_1^H\;\cdots\;\mathbf{x}_{M_c}^H]^H$ and the relationship between $\mathbf{x}$ and $\tilde{\mathbf{X}}$ is the same as defined in Section II. The constraint (\ref{p3_2}) in \textbf{P4} is replaced by (\ref{eq_ic3}) to formulate the optimization for this type of multicell coordinated transmission as a single optimization problem.

The intercell coordination for different scenarios with CBPU are summarized in Algorithm \ref{algo_intercell}. The optimization of both coherent transmission and interference coordination are performed at the CBPU, the optimization of fast selection is performed at each BS. However, the selection of BSs (cell-edge user assignment) for fast selection is carried out by the CBPU.

\section{Numerical Results}
\label{SDR_sim}
This section presents numerical results based on Monte Carlo simulations of the proposed optimization techniques to validate our analysis and evaluate the proposed algorithms.

\begin{table}[htp]
\centering 

\begin{threeparttable}
\caption{Comparison of the Computational Complexities }
\label{tab_complexity}
\def\arraystretch{1.6}\tabcolsep=5pt 
    \begin{tabular}{|p{4.7cm}|c|}
    \hline
     Algorithms & Complexity \\\hline\hline
    Massive precoding without PAPR-awareness\tnote{1,3} & $\mathcal{O}(N_x^{0.5} log(1/\epsilon))$ \\\hline 
    PAPR-aware precoding\tnote{1} & $\mathcal{O}(N_x^{4.5} log(1/\epsilon))$ \\\hline
    Coarse \textit{robust} PAPR-aware precoding\tnote{2} & $\mathcal{O}(N_x^{4.5} log(1/\epsilon))$ \\\hline
    Deterministic fine \textit{robust} PAPR-aware precoding\tnote{2} & $\mathcal{O}(N_x^{4.5} log(1/\epsilon))$\\\hline
    Stochastic fine \textit{robust} PAPR-aware precoding\tnote{2} & $\mathcal{O}(N_x^{4.5}N_s^{0.5} log(1/\epsilon))$ \\\hline
    Coherent transmission with PAPR-aware precoding\tnote{1} & $\mathcal{O}(N_x^{4.5} log(1/\epsilon))$ \\\hline
    Fast cell selection with PAPR-aware precoding\tnote{1} & $\mathcal{O}(M_cN_x^{4.5} log(1/\epsilon))$ \\\hline
    Interference coordination with PAPR-aware precoding\tnote{1} & $\mathcal{O}(M_c^{0.5}N_x^{4.5} log(1/\epsilon))$ \\\hline
    \end{tabular}

\begin{tablenotes}\footnotesize
\item[1] Perfect CSI is assumed.
\item[2] Imperfect CSI is assumed.
\item[3] The baseline is formulated as $\min_{\tilde{\mathbf{X}}}\text{Tr}(\tilde{\mathbf{X}})\;\;\;\text{subject to} \;\;\;(\ref{p3_2})-(\ref{p3_4})$.
\end{tablenotes}

\end{threeparttable}

\end{table}

\subsection{Simulation Parameters}
The effects of PAPR-aware precoding are demonstrated by simulating single cell and multicell scenarios. The users are uniformly distributed, the BSs are equipped with up to 100 antennas and serves 10 single antenna users. The distance between BS antenna elements is one wavelength of the carrier. The case of uncoded 16-QAM is used in the simulation. We assume 128 subcarriers and model the multipath channel as a tapped delay line with 10 taps, each modeled as an independent Rayleigh fading channel.  All simulation results are averaged over 100 channel realizations and run with 2000 randomizations ($M=2000$ in Algorithm \ref{algo_random}) for rank reduction. 
\subsection{Complexity Analysis}

We compare the computational complexities of the proposed approaches in a variety of scenarios.
When the dual-scaling interior-point method is applied, the worst-case computational complexity is $\mathcal{O}(\max\{m, n\}^4n^{0.5} log(1/\epsilon))$ \cite{luo2010semidefinite}, where $m$ is the number of linear constraints, $n$ the order of the PSD constraint, and $\epsilon$ is the solution accuracy. Table \ref{tab_complexity} compares the computational complexities where $N_x\triangleq N_tN_c$, $N_s\triangleq M_rN_c$, and $M_c$ is the number of cooperating BSs. The first row of Table \ref{tab_complexity} is the baseline which represents the complexity of massive MIMO precoding without PAPR-awareness by dropping the constraint (\ref{p1_1}) in problem \textbf{P1}. The subsequent rows demonstrate the other proposed approaches, which have increased complexity comparing with the existing single-input single-output (SISO) systems \cite{ochiai2002performance,dinis2004class}. However, it should be noted that the complexity does not assume sparsity or any special structure in the data matrices $\tilde{\mathbf{Q}}_{im}$ and ${\mathbf{G}}$. Besides, unlike to the traditional PAPR reduction scheme, the scheme proposed in this work is specifically designed for the downlink. It utilizes the excessive DoFs in the base station of massive MIMO systems. With the promising applications of deep machine learning in 5G NR, it is expected that a significant amount of parallel computational resources such as graphics processing units (GPUs) will be equipped in 5G NR base stations and elsewhere (distributed processing units at the network edge and centralized processing units clusters in the ‘Cloud-RAN’ data center). Using the base station side cognition for deep learning is more advantageous than using the device side since there is more flexibility, better handling of large data flows and it can quickly update models, etc. The focus of this initial paper is to introduce the technique. The focus of future work will be to reduce the computational complexity by considering the matrix sparsity and better structure the algorithm for parallel processing.   
\subsection{Single-Cell Scenario}
\begin{figure}[h!]
\centering
\includegraphics[width=3.2in]{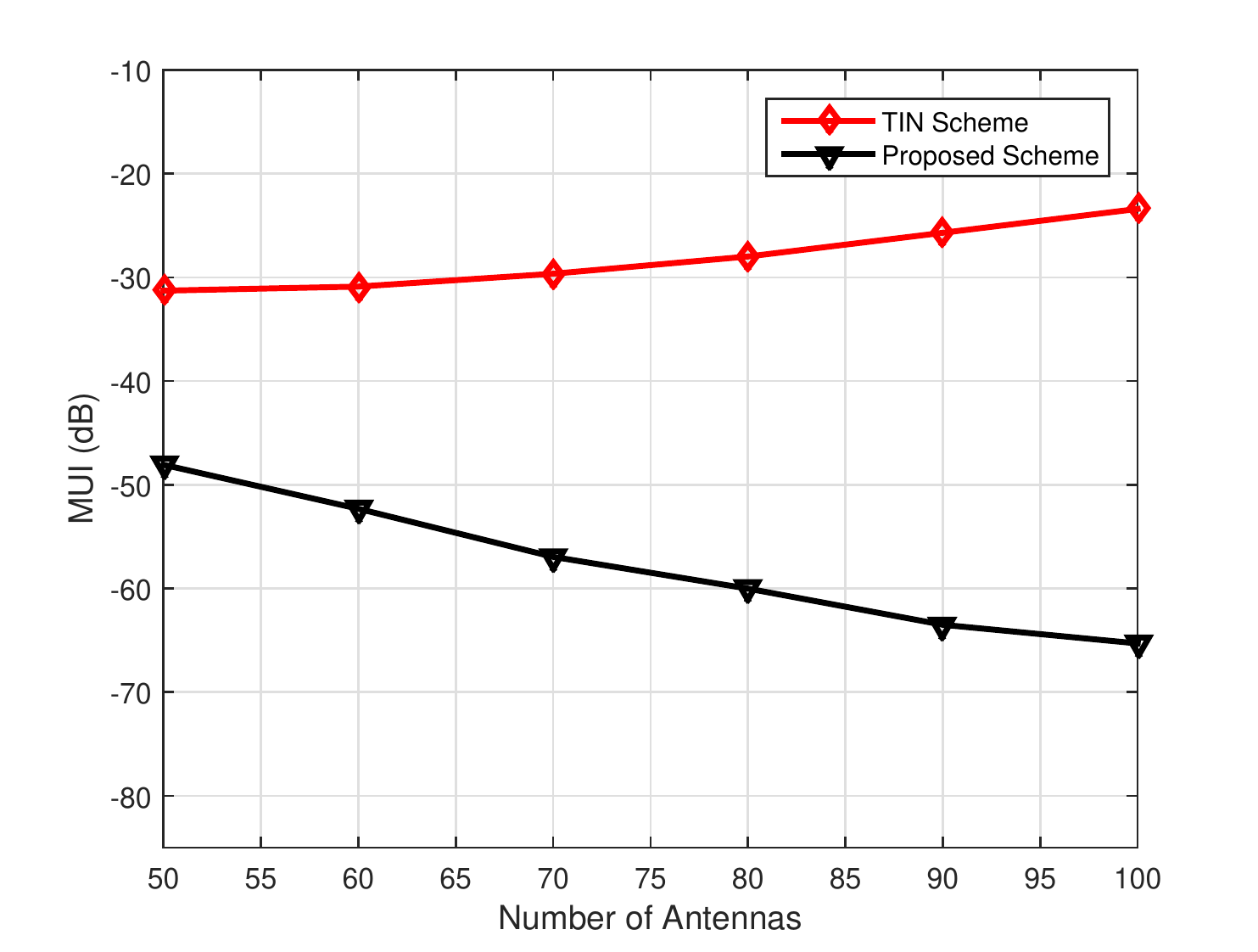}
\caption{Comparison of the MUI for various BS antenna configurations for TIN scheme and the proposed scheme.}
\label{fig_mui1}
\end{figure}

Fig. \ref{fig_mui1} plots the MUI\footnote{The MUI is defined to evaluate the multiuser interference at the receivers that
$\text{MUI}=\frac{\norm{\mathbf{s}-\tilde{\mathbf{H}}\mathbf{x}}_{2}^{2}}{\norm{\mathbf{s}}_{2}^{2}}$.} of the proposed scheme and treating interference as noise (TIN) scheme to evaluate the performance of interference cancellation. The predetermined transmit power is set to be 8 dBm and target PAPR reduction be 6 dB. Besides, the excessive DoFs are also evaluated in the simulation by comparing MUI of the proposed scheme and TIN with different number of BS antennas. The figure shows that the proposed scheme outperforms TIN scheme in terms of MUI and the MUI decreases with the number of antennas at BS since more DoFs are available. 

Fig. \ref{fig_power} compares the average instantaneous transmit power of our proposed method with the joint precoding, modulation and PAPR reduction (PMP) method of \cite{studer2013aware} and the per-antenna constant envelope precoding of \cite{mohammed2013per}. Since the transmit signals are optimized on a symbol-by-symbol basis, the instantaneous transmit power is used to evaluate the performance of the PAPR-aware precoding with different precoding errors. The lower bound of the transmit power without PAPR-awareness which was shown in row 1 of Table \ref{tab_complexity} is also included. As an indicator of MUI allowance, the normalized maximum allowed precoding error (NMAE) used in Fig. \ref{fig_power} is defined as $\alpha =\delta_e/\delta_M$, where $\delta_e$ is the maximum allowed precoding error applied in the optimization and $\delta_M=3.7\times10^{-4}$ the precoding error when the transmit power upper bound is -5 dBm. The simulation results of PAPR-aware precoding with PMP \cite{studer2013aware}\footnote{In \cite{studer2013aware}, the PAPR-aware precoding problem is formulated as $\min_{\tilde{\mathbf{x}}} {\norm{\mathbf{s}-\tilde{\mathbf{H}}\tilde{\mathbf{x}}}_2^2+\lambda\norm{\tilde{\mathbf{x}}}_{\infty}}$ where the signal peak ($\norm{\tilde{\mathbf{x}}}_{\infty}$) instead of PAPR ($N_c\norm{\tilde{\mathbf{x}}}_{\infty}^2 /\norm{\tilde{\mathbf{x}}}_{2}^2$) is applied to realize the convex relaxation, $\lambda\geq 0$ is a regularization parameter to achieve the tradeoff between the PAPR reduction and precoding error (MUI).} and per-antenna constant envelope method \cite{mohammed2013per}\footnote{In \cite{mohammed2013per}, the per-antenna constant envelope method is formulated as $\min_{\theta_1,\cdots,\theta_{N_t}}\sum_{m=1}^{M_r}|\frac{\sum_{n=1}^{N_t} h_{m,n}e^{j\theta_n}}{\sqrt{N_t}}-\sqrt{E_m}u_m|^2$, where $\theta_n$ denotes the phase of the transmit signal of the $n$th antenna, $u_m$ the information symbol of $m$th user, $E_m$ the information symbol energy of $m$th user. Since it is designed for the single-carrier waveform, the parallel OFDM symbols are serialized and per-carrier constant envelope precoding is performed to compare with the performance of the proposed approach.} are included for comparison with the proposed approach. The NMAE values for plotting the performance of the PMP method in Fig. \ref{fig_power} are obtained by altering parameter $\lambda$, whereas the MUI is minimized by constant envelope method to find the optimal phases $\theta_1,\cdots,\theta_{N_t}$ for each subcarrier.

Fig. \ref{fig_power} shows that the instantaneous transmit power of the proposed scheme for different PAPR reduction targets\footnote{The PAPR reduction target is achieved when the PAPR of 99\% of OFDM symbols is reduced with the amount of target.} decreases when the NMAE value is increased, as derived in Section \ref{section3_b}. The achievable lower bound of the transmit power is also plotted. The figure shows that the instantaneous transmit power stays constant for the PMP scheme of \cite{studer2013aware}. The instantaneous transmit power and NMAE which are derived from the optimal phases of the per-antenna constant envelope method are illustrated as a single point in Fig. \ref{fig_power}.

Fig. \ref{fig_sdr_samples} plots the absolute values of the time domain transmit signals for the first antenna. We observe that the transmit signal of the proposed method has a lower dynamic range. The signal generated by the PMP method has only an upper bound but the proposed SDR-based method has both an upper bound and a lower bound. The PMP method aims to reduce the peak power but the proposed SDR-based method aims to reduce the dynamic range of the transmit signals.

\begin{figure}[h!]
\centering
\includegraphics[width=3.2in]{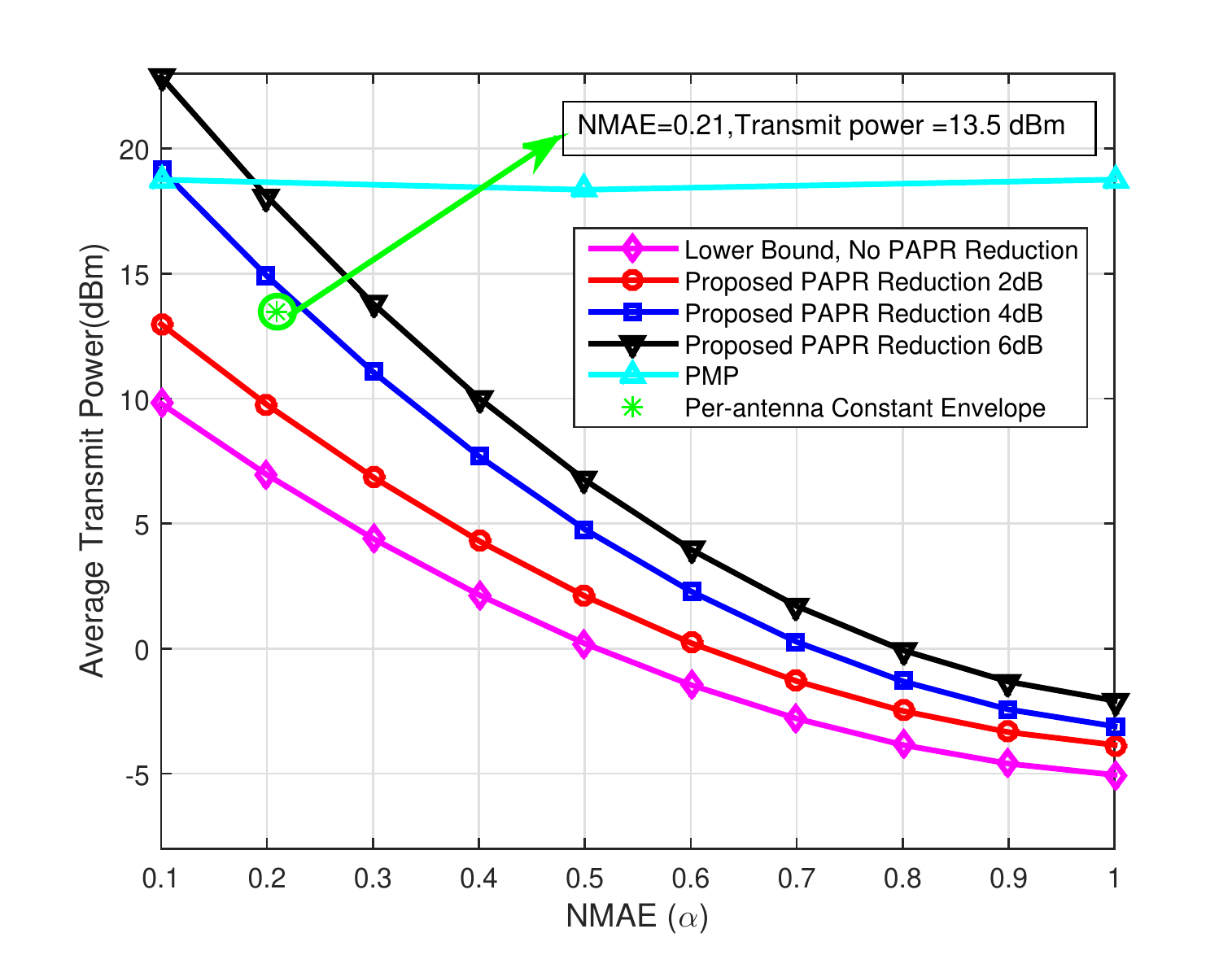}
\caption{Instantaneous transmit power versus NMAE for different PAPR reduction targets, PMP method and constant envelope method (the minimized MUI and associated transmit power of constant envelope method is shown as a single point).}
\label{fig_power}
\end{figure}

Fig. \ref{fig_papr} plots the PAPR CCDFs of the transmit signal for the following approaches: without PAPR-awareness, PMP, clipping and SDR-based with oversampling factor $L$ to be 4. The CCDF curves of all PAPR-aware precoded OFDM signals are shifted to the left with respect to the scheme that does not exploit the DoFs for PAPR reduction. By comparing our method's CCDF with the CCDF of the other approaches, our proposed approach achieves an additional PAPR reduction of approximately 2 dB.

\begin{figure}[h!]
\centering
\includegraphics[width=3.2in]{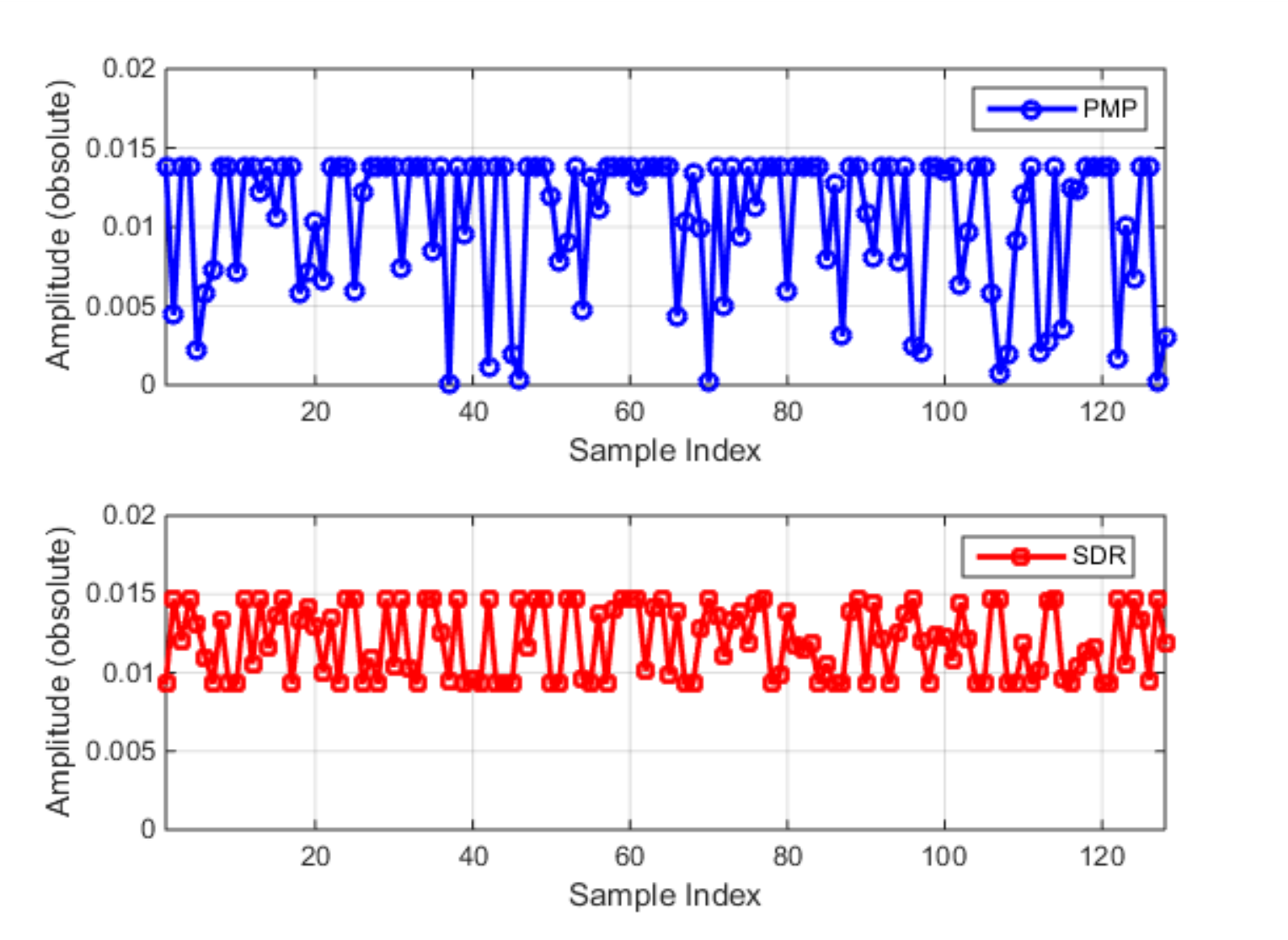}
\caption{The time-domain transmit signals at the first antenna for the PMP method (upper) and the proposed SDR method (lower) ($\delta_e=10^{-4}$).}
\label{fig_sdr_samples}
\end{figure}

\begin{figure}[h!]
\centering
\includegraphics[width=3.2in]{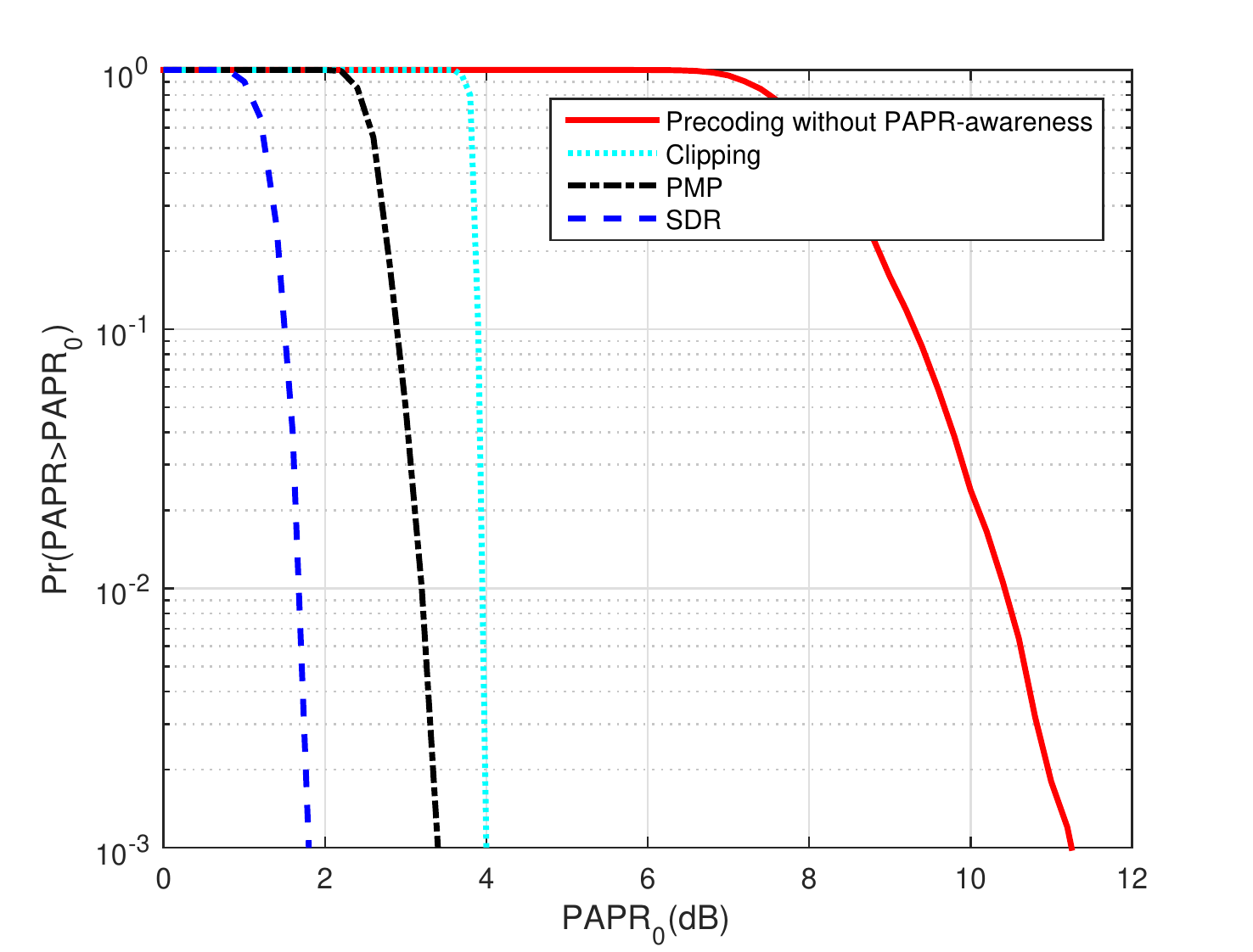}
\caption{Comparison of the PAPR CCDF curves for the first antenna: instantaneous transmit power minimization without PAPR reduction, clipping, PMP and the proposed SDR-based PAPR-aware precoding (oversampling factor $L$ = 4, $\delta_e=10^{-4}$).}
\label{fig_papr}
\end{figure}

The channel estimation error $\Delta\tilde{\mathbf{H}}=\tilde{\mathbf{H}}-\hat{\mathbf{H}}$ is modeled by generating $\Delta\tilde{\mathbf{H}}$ from a zero-mean Gaussian distribution with $\textbf{E}(\Delta\tilde{\mathbf{H}}\Delta\tilde{\mathbf{H}}^H)=\mathbf{R}_{\epsilon}=\sigma_\epsilon^2\mathbf{I}$, where we will use the same $\sigma_\epsilon^2=10^{-3}$ for all users and all subcarriers in the simulation. Moreover, to take different channels into account, the elements of the nominal channel $\hat{\mathbf{H}}$ are randomly generated according to zero-mean, unit-variance, i.i.d. Gaussian distributions. The philosophy of robust PAPR-aware precoding in this paper is to guarantee PAPR reduction and MUI level for any channel realization in the uncertainty region. In other words, we are interested in the behavior of a precoder with uncertainty of CSI.

Fig. \ref{fig_outage} plots the performance of the proposed fine robust precoding strategy with probabilistic uncertainty. 
In order to show the importance of taking the channel uncertainty into account for PAPR-aware precoding, we begin with the non-robust PAPR-aware precoder design when channel uncertainty exists, but the precoder design assumes perfect CSI at the transmitter. Fig. \ref{fig_outage} shows the percentage of constraint violations at different PAPR reduction targets. Note that a larger precoding error allowance $\delta_e$ is used here to achieve PAPR reduction as large as 12 dB. When the channel uncertainty is not considered, the probability of MUI target violation is high, especially with high PAPR reduction targets. However, the constraint violation probabilities are low for our proposed robust PAPR-aware precoding algorithm with different values for $\gamma$. As shown in Fig. \ref{fig_outage}, the constraint violation probability of the proposed robust precoding is approximately 10 times lower than the non-robust precoding when $\gamma=0.02$.

\begin{figure}[h!]
\centering
\includegraphics[width=3.2in]{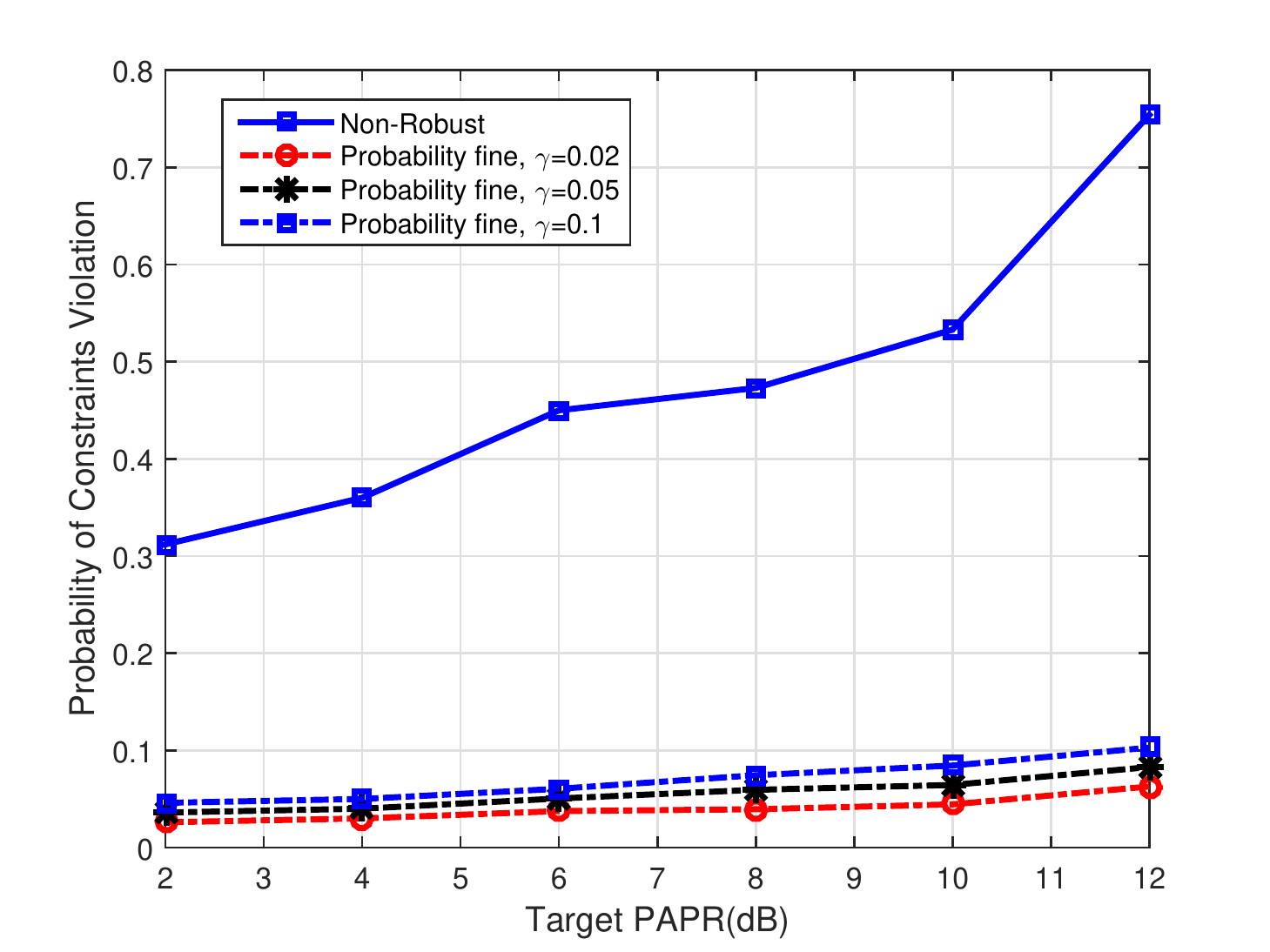}
\caption{Comparison of the percentage of constraints violations for the non-robust precoding \cite{studer2013aware} and fine robust precoding with probabilistic channel uncertainty design ($\delta_e=10^{-3}$).}
\label{fig_outage}
\end{figure} 

An alternative way to investigate the robustness capability is to look at the MUI at the receivers for various error bounds $\epsilon_h$. Fig. \ref{fig_sinr} plots the MUI changes with error bound $\epsilon_h$ for the PMP method, non-robust precoding based on SDR, coarse robust precoding and fine robust precoding with bounded uncertainty design. The results demonstrate that the proposed fine robust precoding using the S-procedure outperforms the coarse robust approach for PAPR reduction. As shown in Fig. \ref{fig_sinr}, the proposed fine robust precoding with bounded uncertainty outperforms the coarse robust precoding, the non-robust precoding based on SDR and the PMP method by 9.2 dB, 15.4 dB, 32.8 dB, respectively when the PAPR reduction is 6 dB. It outperforms the other methods by 7.3 dB, 12.9 dB and 31 dB when the PAPR reduction is 8 dB.

\begin{figure}[h!]
\centering
\includegraphics[width=3.2in]{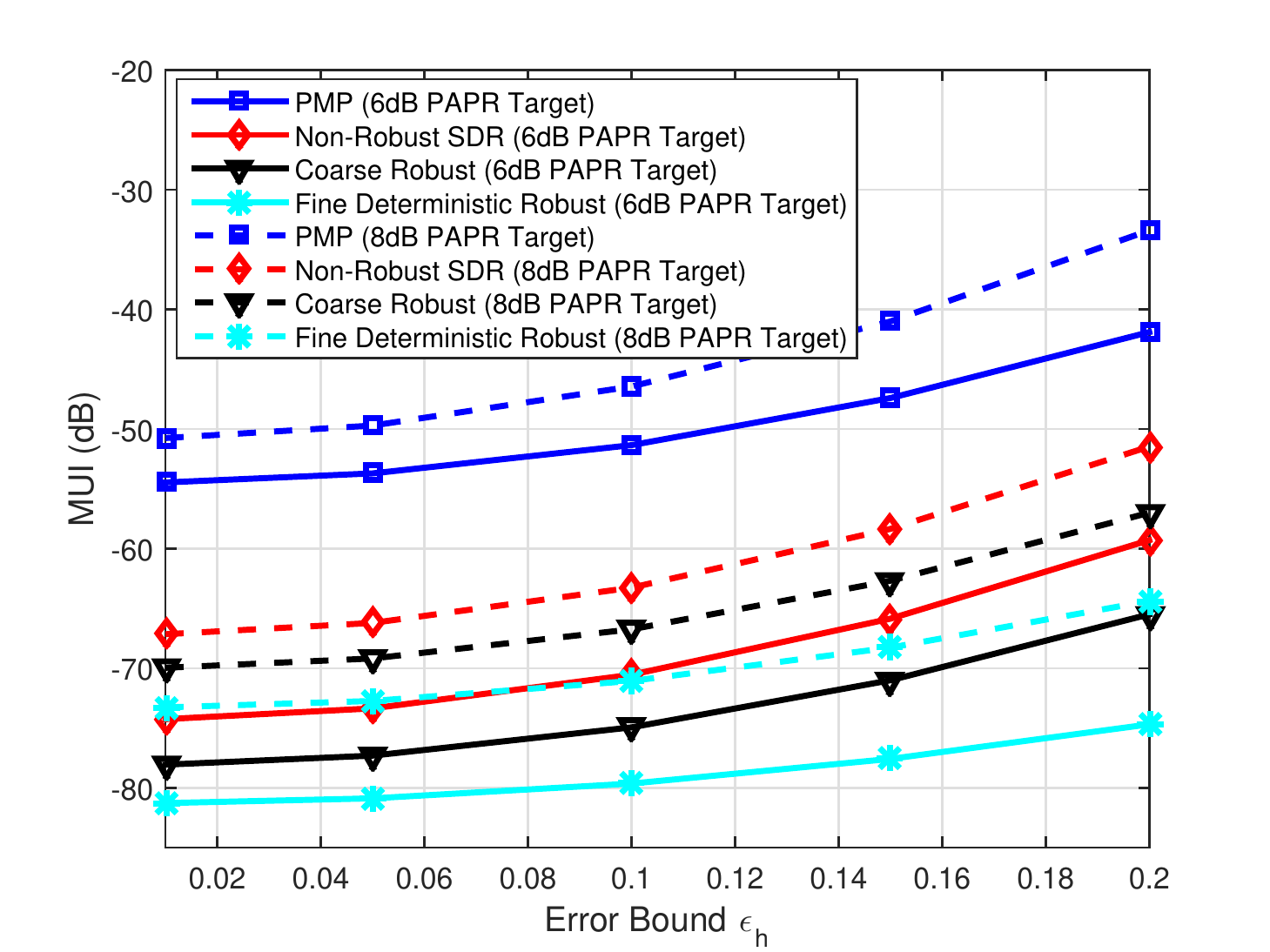}
\caption{Comparison of the MUI for various error bounds $\epsilon_h$ for the PMP method, non-robust precoding, coarse robust precoding and fine robust precoding with bounded channel uncertainty design for different PAPR reduction targets ($\delta_e=10^{-4}$).}
\label{fig_sinr}
\end{figure}

\subsection{Multi-Cell Scenario}
The symbol error rate (SER) performance of the proposed PAPR-aware intercell coordination scheme here is evaluated for a multicell network with 2 cells. Each BS is equipped with 100 antennas. The aim of this numerical analysis is to quantify the benefit of coordinating resource allocation, including scheduling, precoding, and nulling, across multiple cells.

Fig. \ref{fig_ser} plots the worst-case SER versus SNR without considering channel uncertainty (non-robust) for the coherent transmission, fast cell selection and interference coordination precoding scenarios. The PAPR reduction at the BS antennas of both cells are 6 dB. The results show that the coherent transmission outperforms fast cell selection and interference coordination. However, it sacrifices performance in terms of effective transmission rate since both cells coherently transmit the same information symbols. Fast cell selection outperforms the interference coordination in terms of SER. The fine robust precoding with S-procedure results for three different scenarios of intercell coordination are also plotted and show consistently better performance than their non-robust counterparts.
\begin{figure}[h!]
\centering
\includegraphics[width=3.2in]{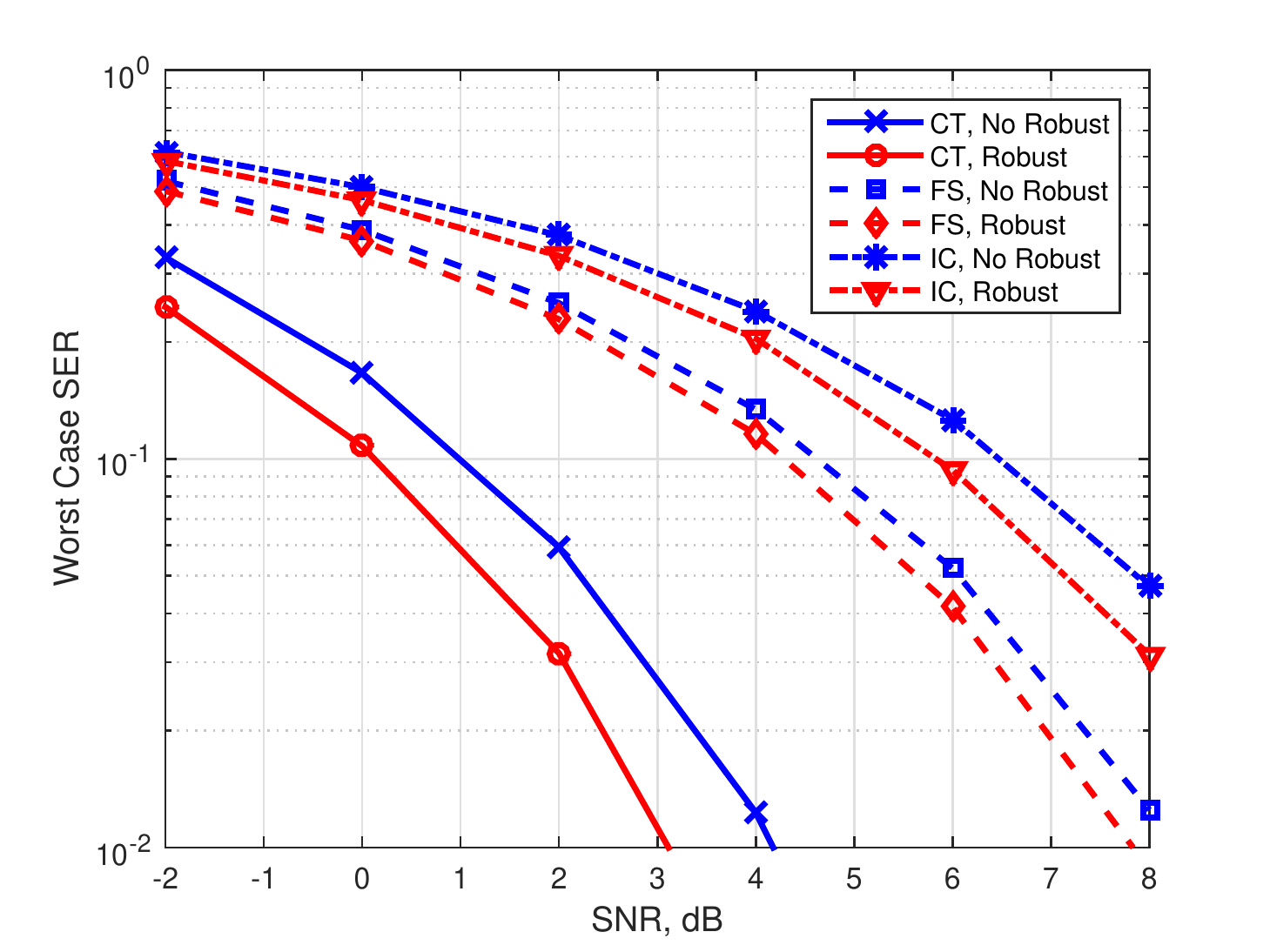}
\caption{Worst-case SER for different scenarios of intercell coordination (16-QAM modulation, $\delta_e=10^{-4}$, PAPR reduction target 6 dB, $\beta_s=0.1$, CT: coherent transmission (20 cell-edge users between the two BSs), FS: fast selection (8 cell-center users randomly located within each cell and 4 cell-edge users between the two BSs), IC: interference coordination (8 cell-center users randomly located within each cell and 4 cell-edge users between the two BSs)).}
\label{fig_ser}
\end{figure}

\section{Conclusion}
This paper has developed efficient SDR-based approaches for PAPR-aware massive MIMO-OFDM systems and has analyzed it in different scenarios of robust precoding and intercell coordination. We have formulated the optimal downlink precoding problem as minimizing the instantaneous transmit power subject to PAPR and MUI constraints. When the number of BS antennas is much larger than the number of users in a massive MIMO system, the proposed SDR-based method exploits the spatial DoFs to yield a per-antenna PAPR-constrained OFDM signal. The randomization based method has been applied for rank reduction of the SDR solution. We have also formulated robust versions to account for channel uncertainty and leverage the SDR method under various CSI uncertainties. Finally, we have developed a PAPR-aware precoding solution for intercell coordination. Using numerical analysis, we demonstrate how our approaches reduce the PAPR of BS antennas for single cell and multiple cells serving cell-center and cell-edge users. Numerical results have been conducted to illustrate the performance and flexibility of the proposed framework.

\section*{Appendix A\\ Proof of Lemma \ref{lemma_1}}
\begin{proof}
Assuming that $\mathbf{s}-\tilde{\mathbf{H}}\mathbf{x}_{opt}=\mathbf{g}$, thus we have $\norm{\mathbf{g}}_2^2\leq\delta_e$ and $\text{Tr}(\mathbf{G}\tilde{\mathbf{X}}_{opt})\leq\delta_e$ since $\mathbf{x}_{opt}$ and $\tilde{\mathbf{X}}_{opt}$ are the optimal values of \textbf{P1} and \textbf{P4} respectively. Note that $\norm{\mathbf{u}}_2^2=1$ and the matrix 
$\mathbf{G}=\begin{bmatrix}
    \tilde{\mathbf{H}}^H       \\
    -\mathbf{s}^H           
\end{bmatrix}\begin{bmatrix}
    \tilde{\mathbf{H}}       & -\mathbf{s}    
\end{bmatrix}$.
The constraint (\ref{p3_2}) with optimal value $\tilde{\mathbf{X}}_{opt}$ can be represented as
\begin{align*}
\text{Tr}(\mathbf{G}\tilde{\mathbf{X}}_{opt})&=\text{Tr}\left(\begin{bmatrix}
    \tilde{\mathbf{H}}^H       \\
    -\mathbf{s}^H           
\end{bmatrix}\begin{bmatrix}
    \tilde{\mathbf{H}}       & -\mathbf{s}    
\end{bmatrix}\begin{bmatrix}
    \mathbf{U^H} \\
    \mathbf{u^H}    
\end{bmatrix}\begin{bmatrix}
    \mathbf{U}       & \mathbf{u}
\end{bmatrix}\right)\\
&=\text{Tr}\left(\begin{bmatrix}
    \tilde{\mathbf{H}}       & -\mathbf{s}    
\end{bmatrix}\begin{bmatrix}
    \mathbf{U^H} \\
    \mathbf{u^H}    
\end{bmatrix}\begin{bmatrix}
    \mathbf{U}       & \mathbf{u}
\end{bmatrix}\begin{bmatrix}
    \tilde{\mathbf{H}}^H       \\
    -\mathbf{s}^H           
\end{bmatrix}\right)\\
&=\norm{\tilde{\mathbf{H}}\mathbf{U^H}-\mathbf{s}\mathbf{u^H}}_2^2\\
&=\norm{\tilde{\mathbf{H}}\mathbf{U^H}-(\tilde{\mathbf{H}}\mathbf{x}_{opt}-\mathbf{g})\mathbf{u^H}}_2^2\\
&\geq\norm{\tilde{\mathbf{H}}(\mathbf{U^H}-\mathbf{x}_{opt}\mathbf{u^H})}_2^2-\norm{\mathbf{g}\mathbf{u^H}}_2^2\\
&\geq\norm{\Delta\mathbf{x}_u}_2^2|\lambda_{min}(\tilde{\mathbf{H}}\tilde{\mathbf{H}}^H)|-\delta_e.
\end{align*}
As a result, we can derive that $\norm{\Delta\mathbf{x}_u}_2^2\leq\frac{2\delta_e}{|\lambda_{min}({\mathbf{H}}{\mathbf{H}}^H)|}$, note that $\tilde{\mathbf{H}}\tilde{\mathbf{H}}^H={\mathbf{H}}\bar{\mathbf{P}}\bar{\mathbf{Q}}\bar{\mathbf{Q}}^H\bar{\mathbf{P}}^H{\mathbf{H}}^H=\mathbf{H}\mathbf{H}^H$, and $\mathbf{H}=\text{diag}\{\mathbf{H}_1,\cdots,\mathbf{H}_{N_c}\}$.
\end{proof}

\section*{Appendix B\\Proof of Lemma \ref{lemma_deter}}
\begin{proof}
The constraint of (\ref{con_s_procedure}) is equivalent to 
\begin{align}\nonumber
&\text{Tr}\left((\hat{\mathbf{H}}_s+\Delta \mathbf{H}_s)\mathbf{W}_x(\hat{\mathbf{H}}_s+\Delta \mathbf{H}_s)^H\right)\\\nonumber
=&\text{Tr}\left(\hat{\mathbf{H}}_s\mathbf{W}_x\hat{\mathbf{H}}_s^H\right)+\text{Tr}\left(\Delta \mathbf{H}_s\mathbf{W}_x\Delta \mathbf{H}_s^H\right)\\\nonumber
&\quad\quad\quad\quad\quad\quad\quad\quad\quad\quad+2\text{Re}\left(\text{Tr}(\hat{\mathbf{H}}_s\mathbf{W}_x\Delta\mathbf{H}_s^H)\right)\\\nonumber
=&\text{Tr}\left(\hat{\mathbf{H}}_s\mathbf{W}_x\hat{\mathbf{H}}_s^H\right)+\Delta\mathbf{h}_s^H(\mathbf{I}_{N_s}\otimes\mathbf{W}_x)\Delta\mathbf{h}_s\\
&\quad\quad\quad\quad\quad\quad\quad\quad\quad\quad+2\text{Re}\left(\text{vec}(\hat{\mathbf{H}}_s\mathbf{W}_x)^H\Delta\mathbf{h}_s)\right)\leq\delta_e\label{sdphwh},
\end{align}
where $\Delta\mathbf{h}_s=\text{vec}(\Delta \mathbf{H}_s)$. Here we use the properties of the trace operator $\text{Tr}(\mathbf{A}\mathbf{B}\mathbf{A}^H)=\text{vec}(\mathbf{A})(\mathbf{I}\otimes\mathbf{B})\text{vec}(\mathbf{A})^H$ and $\text{Tr}(\mathbf{A}\mathbf{B}^H)=\text{vec}(\mathbf{A})\text{vec}(\mathbf{B})^H$ \cite{zhang2004matrix}. Due to the property of the column-by-column matrix vectorization and (\ref{upsilon_d}), we have 
\begin{equation}
\norm{\Delta\mathbf{h}_s}_2^2\leq\epsilon_h^2\label{deltah}
\end{equation}
from (\ref{upsilon_d}). In order to further derive the precoding error constraint under the scenario of channel uncertainty, we adapt a S-procedure lemma from  \cite{boyd2004convex} to find the alternative of quadratic inequalities (\ref{sdphwh}) and (\ref{deltah}).
\begin{lemma}
Consider $\mathbf{A}_1,\mathbf{A}_2\in\mathbb{H}^{n\times n}$, $\mathbf{b}_1,\mathbf{b}_2\in\mathbb{C}^{n}$, $c_1,c_2\in\mathbb{R}$, and suppose there exists an $\mathbf{u}_0$ with \cite{boyd2004convex}
\begin{equation}
\mathbf{u}_0^H\mathbf{A}_2\mathbf{u}_0+2\text{Re}(\mathbf{b}_2^H\mathbf{u}_0)+c_2<0\nonumber.
\end{equation}
Then, the inequality 
\begin{align}\nonumber
&\mathbf{u}^H\mathbf{A}_1\mathbf{u}+2\text{Re}(\mathbf{b}_1^H\mathbf{u})+c_1\geq 0,\\
\forall\mathbf{u}\in&\Phi_u:\left\{\mathbf{u}|\mathbf{u}^H\mathbf{A}_2\mathbf{u}+2\text{Re}(\mathbf{b}_2^H\mathbf{u})+c_2\leq 0\right\}
\end{align}
holds if and only if there exists $\lambda\geq 0$ such that
\begin{equation}
\begin{bmatrix}
    \mathbf{A}_1+\lambda\mathbf{A}_2 &  \mathbf{b}_1+\mathbf{b}_2    \\
    (\mathbf{b}_1+\mathbf{b}_2)^H       &  c_1+\lambda c_2
\end{bmatrix}
\succeq \mathbf{0}.
\end{equation}\label{lemmasp}
\end{lemma}
Since $\mathbf{W}_x=\tilde{\mathbf{X}}|_{t=1}$, by applying Lemma \ref{lemmasp}, there exists $\lambda\geq 0$ that the quadratic inequality constraints with considering the channel uncertainty in (\ref{sdphwh}) and (\ref{deltah}) can reformulated as
\begin{equation}
\begin{bmatrix}
    -\mathbf{I}_{N_s}\otimes\tilde{\mathbf{X}}+\lambda\mathbf{I}_{Ns(N_x+1)} &  -\text{vec}(\hat{\mathbf{H}}_s\tilde{\mathbf{X}})    \\
    -\text{vec}(\hat{\mathbf{H}}_s\tilde{\mathbf{X}})^H       &  \delta_e-\text{Tr}\left(\hat{\mathbf{H}}_s\tilde{\mathbf{X}}\hat{\mathbf{H}}_s^H\right)-\lambda\epsilon_h^2
\end{bmatrix}
\succeq \mathbf{0}\label{lmi},
\end{equation}
where $\lambda$ is a variable which needs to be found in the optimization. By applying Schur's complement \cite{horn2012matrix}, the equation (\ref{lmi}) can be further represented as
\begin{align}\nonumber
-\text{vec}(\hat{\mathbf{H}}_s\tilde{\mathbf{X}})^H&\left(-\mathbf{I}_{N_s}\otimes\tilde{\mathbf{X}}+\lambda\mathbf{I}_{Ns(N_x+1)}\right)^{-1}\text{vec}(\hat{\mathbf{H}}_s\tilde{\mathbf{X}})\\
&+\delta_e-\lambda\epsilon_h^2-\text{Tr}\left(\hat{\mathbf{H}}_s\tilde{\mathbf{X}}\hat{\mathbf{H}}_s^H\right)\geq 0.
\end{align}
Thus we have
\begin{align}\nonumber
\delta_e-\lambda\epsilon_h^2-&\text{Tr}\left(\tilde{\mathbf{X}}\hat{\mathbf{H}}_s^H\hat{\mathbf{H}}_s\right)\\
&-\text{Tr}\left(\tilde{\mathbf{X}}(-\tilde{\mathbf{X}}+\lambda\mathbf{I}_{N_x+1})^{-1} \tilde{\mathbf{X}}\hat{\mathbf{H}}_s^H\hat{\mathbf{H}}_s\right)\geq 0,
\end{align}
which can be relaxed as (\ref{x-y}) and (\ref{succ_y}).
\end{proof}

\section*{Appendix C\\Proof of Lemma \ref{lemma_statistical}}
\begin{proof}
Since $\text{vec}(\Delta\tilde{\mathbf{H}})\sim \mathcal{CN}(\mathbf{0},\mathbf{R}_{\epsilon})$, we have that $\Delta\mathbf{h}_s\sim \mathcal{CN}(\mathbf{0},\mathbf{R}_{\epsilon}')$, where $\mathbf{R}_{\epsilon}'=[\mathbf{R}_{\epsilon},\mathbf{0}]$. Through the following conversion:
\begin{equation}
\Delta\mathbf{h}_s=\mathbf{R}_{\epsilon}^{'1/2}\mathbf{e}_h,
\end{equation}
where $\mathbf{e}_h\sim \mathcal{CN}(\mathbf{0},\mathbf{I})$, the entries of $\mathbf{R}_{\epsilon}^{'1/2}$ are the square root of the entries of $\mathbf{R}_{\epsilon}'$. The probabilistic inequality (\ref{bti}) can be transformed as
\begin{align}\nonumber
\text{Pr}\Big(\mathbf{e}_h^H\mathbf{W}_{\epsilon}\mathbf{e}_h+2\text{Re}\left(\mathbf{g}_{\epsilon}^H\mathbf{e}_h\right)+\delta_e-\text{Tr}\left(\hat{\mathbf{H}}_s\tilde{\mathbf{X}}\hat{\mathbf{H}}_s^H\right)&\geq 0\Big)\\\nonumber
&\geq 1-\gamma,
\end{align}
where we defined $\mathbf{W}_{\epsilon}$ and $\mathbf{g}_{\epsilon}$ as follows
\begin{align}
\mathbf{W}_{\epsilon}&\triangleq-{\mathbf{R}_{\epsilon}^{'1/2}}(\mathbf{I}_{N_s}\otimes\mathbf{W}_x)\mathbf{R}_{\epsilon}^{'1/2}\\
\mathbf{g}_{\epsilon}&\triangleq-{\mathbf{R}_{\epsilon}^{'1/2}}\text{vec}(\hat{\mathbf{H}}_s\mathbf{W}_x),
\end{align}
since $\mathbf{R}_{\epsilon}'$ is a Hermitian matrix, note that $\mathbf{W}_x=\tilde{\mathbf{X}}|_{t=1}$.
\begin{lemma}\label{lemma_bti}
Consider $\mathbf{W}\in\mathbb{H}^{n\times n}$, $\mathbf{g}\in\mathbb{C}^{n}$, $c\in\mathbb{R}$, $\mathbf{u}\sim\mathcal{CN}(\mathbf{0},\mathbf{I})$ and $\gamma\in(0,1]$ and define $\theta\triangleq-\log(\gamma)$ \cite{bechar2009bernstein}. The following condition 
\begin{equation}
\text{Pr}\left\{\mathbf{u}^H\mathbf{W}\mathbf{u}+2\text{Re}(\mathbf{g}^H\mathbf{u})+c\geq 0\right\}\geq 1-\gamma
\end{equation}
is satisfied if and only if all of the following inequalities hold:
\end{lemma}
\begin{subequations}
\begin{align}
&\text{Tr}(\mathbf{W})-\sqrt{-2\theta}\lambda_1-\theta \lambda_2+c\geq 0\\
&\sqrt{\norm{\mathbf{W}}_2^2+2\norm{\mathbf{g}}_2^2}\leq \lambda_1\\
&\lambda_2\mathbf{I}+\mathbf{W}\succeq\mathbf{0}\\
&\lambda_2\geq 0,
\end{align}
\end{subequations}
where $\lambda_1,\lambda_2\in\mathbb{R}^+$ are slack variables.
By applying Lemma \ref{lemma_bti}, we obtain the constraints of (\ref{bti_c_1})-(\ref{bti_c_4}).
\end{proof}

\bibliographystyle{IEEEtran}
%

%




\bibliographystyle{ieeetr}
\bibliography{myreference}

\begin{thebibliography}{10}
\providecommand{\url}[1]{#1}
\csname url@samestyle\endcsname
\providecommand{\newblock}{\relax}
\providecommand{\bibinfo}[2]{#2}
\providecommand{\BIBentrySTDinterwordspacing}{\spaceskip=0pt\relax}
\providecommand{\BIBentryALTinterwordstretchfactor}{4}
\providecommand{\BIBentryALTinterwordspacing}{\spaceskip=\fontdimen2\font plus
\BIBentryALTinterwordstretchfactor\fontdimen3\font minus
  \fontdimen4\font\relax}
\providecommand{\BIBforeignlanguage}[2]{{%
\expandafter\ifx\csname l@#1\endcsname\relax
\typeout{** WARNING: IEEEtran.bst: No hyphenation pattern has been}%
\typeout{** loaded for the language `#1'. Using the pattern for}%
\typeout{** the default language instead.}%
\else
\language=\csname l@#1\endcsname
\fi
#2}}
\providecommand{\BIBdecl}{\relax}
\BIBdecl

\bibitem{access2014physical}
E.~U. T.~R. Access, ``{Physical channels and modulation (3GPP TS 36.211 version
  12.4. 0 Release 12),”},'' \emph{ETSI, Standard}, vol.~12, no.~0, 2014.

\bibitem{ochiai2000performance}
H.~Ochiai and H.~Imai, ``Performance of the deliberate clipping with adaptive
  symbol selection for strictly band-limited {OFDM} systems,'' \emph{IEEE
  Journal on Selected Areas in Communications}, vol.~18, no.~11, pp.
  2270--2277, 2000.

\bibitem{ren2003complementary}
G.~Ren, H.~Zhang, and Y.~Chang, ``A complementary clipping transform technique
  for the reduction of peak-to-average power ratio of {OFDM} system,''
  \emph{IEEE Transactions on Consumer Electronics}, vol.~49, no.~4, pp.
  922--926, 2003.

\bibitem{slimane2007reducing}
S.~B. Slimane, ``Reducing the peak-to-average power ratio of {OFDM} signals
  through precoding,'' \emph{IEEE Transactions on Vehicular Technology},
  vol.~56, no.~2, pp. 686--695, 2007.

\bibitem{ITU21meeting}
ITU, ``Report on the twenty-first meeting of working party {5D},'' 2015.

\bibitem{ngo2013energy}
H.~Q. Ngo, E.~G. Larsson, and T.~L. Marzetta, ``Energy and spectral efficiency
  of very large multiuser {MIMO} systems,'' \emph{IEEE Transactions on
  Communications}, vol.~61, no.~4, pp. 1436--1449, 2013.

\bibitem{mohammed2013per}
S.~K. Mohammed and E.~G. Larsson, ``Per-antenna constant envelope precoding for
  large multi-user {MIMO} systems,'' \emph{IEEE Transactions on
  Communications}, vol.~61, no.~3, pp. 1059--1071, 2013.

\bibitem{zhang2016per}
J.~Zhang, Y.~Huang, J.~Wang, B.~Ottersten, and L.~Yang, ``Per-antenna constant
  envelope precoding and antenna subset selection: A geometric approach,''
  \emph{IEEE Transactions on Signal Processing}, vol.~64, no.~23, pp.
  6089--6104, 2016.

\bibitem{pan2014constant}
J.~Pan and W.-K. Ma, ``Constant envelope precoding for single-user large-scale
  {MISO} channels: efficient precoding and optimal designs,'' \emph{IEEE
  Journal of Selected Topics in Signal Processing}, vol.~8, no.~5, pp.
  982--995, 2014.

\bibitem{chen2016low}
J.-C. Chen, C.-J. Wang, K.-K. Wong, and C.-K. Wen, ``Low-complexity precoding
  design for massive multiuser {MIMO} systems using approximate message
  passing,'' \emph{IEEE Transactions on Vehicular Technology}, vol.~65, no.~7,
  pp. 5707--5714, 2016.

\bibitem{studer2013aware}
C.~Studer and E.~G. Larsson, ``{PAR-aware} large-scale multi-user {MIMO-OFDM}
  downlink,'' \emph{IEEE Journal on Selected Areas in Communications}, vol.~31,
  no.~2, pp. 303--313, 2013.

\bibitem{cha2014generalized}
H.-S. Cha, H.~Chae, K.~Kim, J.~Jang, J.~Yang, and D.~K. Kim, ``Generalized
  inverse aided {PAPR}-aware linear precoder design for{ MIMO-OFDM} system,''
  \emph{IEEE Communications Letters}, vol.~18, no.~8, pp. 1363--1366, 2014.

\bibitem{bao2016efficient}
H.~Bao, J.~Fang, Z.~Chen, H.~Li, and S.~Li, ``An efficient bayesian {PAPR}
  reduction method for {OFDM}-based massive {MIMO} systems,'' \emph{IEEE
  Transactions on Wireless Communications}, vol.~15, no.~6, pp. 4183--4195,
  2016.

\bibitem{huang2010rank}
Y.~Huang and D.~P. Palomar, ``Rank-constrained separable semidefinite
  programming with applications to optimal beamforming,'' \emph{IEEE
  Transactions on Signal Processing}, vol.~58, no.~2, pp. 664--678, 2010.

\bibitem{wiesel2005semidefinite}
A.~Wiesel, Y.~C. Eldar, and S.~Shamai, ``Semidefinite relaxation for detection
  of 16-{QAM} signaling in {MIMO} channels,'' \emph{IEEE Signal Processing
  Letters}, vol.~12, no.~9, pp. 653--656, 2005.

\bibitem{mobasher2007near}
A.~Mobasher, M.~Taherzadeh, R.~Sotirov, and A.~K. Khandani, ``A
  near-maximum-likelihood decoding algorithm for {MIMO} systems based on
  semi-definite programming,'' \emph{IEEE Transactions on Information Theory},
  vol.~53, no.~11, pp. 3869--3886, 2007.

\bibitem{wang2011papr}
Y.-C. Wang, J.-L. Wang, K.-C. Yi, and B.~Tian, ``{PAPR} reduction of {OFDM}
  signals with minimized {EVM} via semidefinite relaxation,'' \emph{IEEE
  Transactions on Vehicular Technology}, vol.~60, no.~9, pp. 4662--4667, 2011.

\bibitem{zhang2014papr}
H.~Zhang, Y.~Yuan, and W.~Xu, ``{PAPR} reduction for {DCO-OFDM} visible light
  communications via semidefinite relaxation,'' \emph{IEEE Photonics Technology
  Letters}, vol.~26, no.~17, pp. 1718--1721, 2014.

\bibitem{ma2004semidefinite}
W.-K. Ma, P.-C. Ching, and Z.~Ding, ``Semidefinite relaxation based multiuser
  detection for {M-ary} {PSK} multiuser systems,'' \emph{IEEE Transactions on
  Signal Processing}, vol.~52, no.~10, pp. 2862--2872, 2004.

\bibitem{ma2002quasi}
W.-K. Ma, T.~N. Davidson, K.~M. Wong, Z.-Q. Luo, and P.-C. Ching,
  ``Quasi-maximum-likelihood multiuser detection using semi-definite relaxation
  with application to synchronous {CDMA},'' \emph{IEEE Transactions on Signal
  Processing}, vol.~50, no.~4, pp. 912--922, 2002.

\bibitem{goemans1995improved}
M.~X. Goemans and D.~P. Williamson, ``Improved approximation algorithms for
  maximum cut and satisfiability problems using semidefinite programming,''
  \emph{Journal of the ACM (JACM)}, vol.~42, no.~6, pp. 1115--1145, 1995.

\bibitem{sidiropoulos2006semidefinite}
N.~D. Sidiropoulos and Z.-Q. Luo, ``A semidefinite relaxation approach to
  {MIMO} detection for high-order {QAM} constellations,'' \emph{IEEE Signal
  Processing Letters}, vol.~13, no.~9, pp. 525--528, 2006.

\bibitem{yao2017sustainable}
M.~Yao, M.~M. Sohul, X.~Ma, V.~Marojevic, and J.~H. Reed, ``Sustainable green
  networking: exploiting degrees of freedom towards energy-efficient 5g
  systems,'' \emph{Wireless Networks}, pp. 1--10, 2017.

\bibitem{ochiai2002performance}
H.~Ochiai and H.~Imai, ``{Performance analysis of deliberately clipped OFDM
  signals},'' \emph{IEEE Transactions on communications}, vol.~50, no.~1, pp.
  89--101, 2002.

\bibitem{gazor2012tone}
S.~Gazor and R.~AliHemmati, ``Tone reservation for {OFDM} systems by maximizing
  signal-to-distortion ratio,'' \emph{IEEE Transactions on Wireless
  Communications}, vol.~11, no.~2, pp. 762--770, 2012.

\bibitem{aggarwal2006minimizing}
A.~Aggarwal and T.~H. Meng, ``Minimizing the peak-to-average power ratio of
  {OFDM} signals using convex optimization,'' \emph{IEEE Transactions on Signal
  Processing}, vol.~54, no.~8, pp. 3099--3110, 2006.

\bibitem{alavi2005papr}
A.~Alavi, C.~Tellambura, and I.~Fair, ``{PAPR} reduction of {OFDM} signals
  using partial transmit sequence: an optimal approach using sphere decoding,''
  \emph{IEEE Communications Letters}, vol.~9, no.~11, pp. 982--984, 2005.

\bibitem{chen2010partial}
J.-C. Chen, ``Partial transmit sequences for {PAPR} reduction of {OFDM} signals
  with stochastic optimization techniques,'' \emph{IEEE Transactions on
  Consumer Electronics}, vol.~56, no.~3, pp. 1229--1234, 2010.

\bibitem{breiling2001slm}
M.~Breiling, S.~H. M{\"u}ller-Weinfurtner, and J.~B. Huber, ``{SLM} peak-power
  reduction without explicit side information,'' \emph{IEEE Communications
  Letters}, vol.~5, no.~6, pp. 239--241, 2001.

\bibitem{vucic2009robust}
N.~Vucic, H.~Boche, and S.~Shi, ``Robust transceiver optimization in downlink
  multiuser {MIMO} systems,'' \emph{IEEE Transactions on Signal Processing},
  vol.~57, no.~9, pp. 3576--3587, 2009.

\bibitem{tajer2011robust}
A.~Tajer, N.~Prasad, and X.~Wang, ``Robust linear precoder design for
  multi-cell downlink transmission,'' \emph{IEEE Transactions on Signal
  Processing}, vol.~59, no.~1, pp. 235--251, 2011.

\bibitem{shen2012distributed}
C.~Shen, T.-H. Chang, K.-Y. Wang, Z.~Qiu, and C.-Y. Chi, ``Distributed robust
  multicell coordinated beamforming with imperfect {CSI}: An {ADMM} approach,''
  \emph{IEEE Transactions on Signal Processing}, vol.~60, no.~6, pp.
  2988--3003, 2012.

\bibitem{zheng2009robust}
G.~Zheng, K.-K. Wong, and B.~Ottersten, ``Robust cognitive beamforming with
  bounded channel uncertainties,'' \emph{IEEE Transactions on Signal
  Processing}, vol.~57, no.~12, pp. 4871--4881, 2009.

\bibitem{wang2011robust}
J.~Wang, G.~Scutari, and D.~P. Palomar, ``Robust {MIMO} cognitive radio via
  game theory,'' \emph{IEEE Transactions on Signal Processing}, vol.~59, no.~3,
  pp. 1183--1201, 2011.

\bibitem{gharavol2010robust}
E.~A. Gharavol, Y.-C. Liang, and K.~Mouthaan, ``Robust downlink beamforming in
  multiuser {MISO} cognitive radio networks with imperfect channel-state
  information,'' \emph{IEEE Transactions on Vehicular Technology}, vol.~59,
  no.~6, pp. 2852--2860, 2010.

\bibitem{rong2006robust}
Y.~Rong, S.~A. Vorobyov, and A.~B. Gershman, ``Robust linear receivers for
  multiaccess space-time block-coded {MIMO} systems: A probabilistically
  constrained approach,'' \emph{IEEE Journal on Selected Areas in
  Communications}, vol.~24, no.~8, pp. 1560--1570, 2006.

\bibitem{chalise2007robust}
B.~K. Chalise, S.~Shahbazpanahi, A.~Czylwik, and A.~B. Gershman, ``Robust
  downlink beamforming based on outage probability specifications,'' \emph{IEEE
  Transactions on Wireless Communications}, vol.~6, no.~10, 2007.

\bibitem{wang2011probabilistic}
K.-Y. Wang, T.-H. Chang, W.-K. Ma, A.~M.-C. So, and C.-Y. Chi, ``Probabilistic
  sinr constrained robust transmit beamforming: A bernstein-type inequality
  based conservative approach,'' in \emph{Acoustics, Speech and Signal
  Processing (ICASSP), 2011 IEEE International Conference on}.\hskip 1em plus
  0.5em minus 0.4em\relax IEEE, 2011, pp. 3080--3083.

\bibitem{chung2011probabilistic}
P.-J. Chung, H.~Du, and J.~Gondzio, ``A probabilistic constraint approach for
  robust transmit beamforming with imperfect channel information,'' \emph{IEEE
  Transactions on Signal Processing}, vol.~59, no.~6, pp. 2773--2782, 2011.

\bibitem{huang2011distributed}
Y.~Huang, G.~Zheng, M.~Bengtsson, K.-K. Wong, L.~Yang, and B.~Ottersten,
  ``Distributed multicell beamforming with limited intercell coordination,''
  \emph{IEEE Transactions on Signal Processing}, vol.~59, no.~2, pp. 728--738,
  2011.

\bibitem{sawahashi2010coordinated}
M.~Sawahashi, Y.~Kishiyama, A.~Morimoto, D.~Nishikawa, and M.~Tanno,
  ``Coordinated multipoint transmission/reception techniques for {LTE}-advanced
  [coordinated and distributed {MIMO}],'' \emph{IEEE Wireless Communications},
  vol.~17, no.~3, 2010.

\bibitem{choi2008capacity}
W.~Choi and J.~G. Andrews, ``The capacity gain from intercell scheduling in
  multi-antenna systems,'' \emph{IEEE Transactions on Wireless Communications},
  vol.~7, no.~2, 2008.

\bibitem{somekh2009cooperative}
O.~Somekh, O.~Simeone, Y.~Bar-Ness, A.~M. Haimovich, and S.~Shamai,
  ``Cooperative multicell zero-forcing beamforming in cellular downlink
  channels,'' \emph{IEEE Transactions on Information Theory}, vol.~55, no.~7,
  pp. 3206--3219, 2009.

\bibitem{gomez2012link}
I.~Gomez-Miguelez, V.~Marojevic, and A.~Gelonch, ``{Link adaptation for
  energy-efficient transmission with receiver CSI},'' \emph{IEEE Communications
  Letters}, vol.~16, no.~9, pp. 1412--1415, 2012.

\bibitem{ochiai2012instantaneous}
H.~Ochiai, ``{On instantaneous power distributions of single-carrier FDMA
  signals},'' \emph{IEEE Wireless Communications Letters}, vol.~1, no.~2, pp.
  73--76, 2012.

\bibitem{yao2018dpd}
M.~Yao, M.~Sohul, R.~Nealy, V.~Marojevic, and J.~Reed, ``{A Digital
  Predistortion Scheme Exploiting Degrees-of-Freedom for Massive MIMO
  Systems},'' in \emph{2018 IEEE International Conference on Communications
  (ICC)}.\hskip 1em plus 0.5em minus 0.4em\relax IEEE, 2018.

\bibitem{vuk2018}
V.~Marojevic, I.~Gomez-Miguelez, and A.~Gelonch-Bosch, ``Analysis of wireless
  transceiver power tradeoffs for green communications,'' \emph{Electronics
  Letters}, December 2017.

\bibitem{ochiai2003performance}
H.~Ochiai, ``{Performance analysis of peak power and band-limited OFDM system
  with linear scaling},'' \emph{IEEE Transactions on Wireless Communications},
  vol.~2, no.~5, pp. 1055--1065, 2003.

\bibitem{ochiai2001distribution}
H.~Ochiai and H.~Imai, ``On the distribution of the peak-to-average power ratio
  in {OFDM} signals,'' \emph{IEEE Transactions on Communications}, vol.~49,
  no.~2, pp. 282--289, 2001.

\bibitem{tellambura2001computation}
C.~Tellambura, ``{Computation of the continuous-time PAR of an OFDM signal with
  BPSK subcarriers},'' \emph{IEEE Communications letters}, vol.~5, no.~5, pp.
  185--187, 2001.

\bibitem{michailow2013low}
N.~Michailow and G.~Fettweis, ``Low peak-to-average power ratio for next
  generation cellular systems with generalized frequency division
  multiplexing,'' in \emph{Intelligent Signal Processing and Communications
  Systems (ISPACS), 2013 International Symposium on}.\hskip 1em plus 0.5em
  minus 0.4em\relax IEEE, 2013, pp. 651--655.

\bibitem{bjornson2014massive}
E.~Bj{\"o}rnson, J.~Hoydis, M.~Kountouris, and M.~Debbah, ``Massive {MIMO}
  systems with non-ideal hardware: Energy efficiency, estimation, and capacity
  limits,'' \emph{IEEE Transactions on Information Theory}, vol.~60, no.~11,
  pp. 7112--7139, 2014.

\bibitem{yu2007transmitter}
W.~Yu and T.~Lan, ``Transmitter optimization for the multi-antenna downlink
  with per-antenna power constraints,'' \emph{IEEE Transactions on Signal
  Processing}, vol.~55, no.~6, pp. 2646--2660, 2007.

\bibitem{boyd2004convex}
S.~Boyd and L.~Vandenberghe, \emph{Convex optimization}.\hskip 1em plus 0.5em
  minus 0.4em\relax Cambridge University Press, 2004.

\bibitem{luo2010semidefinite}
Z.-Q. Luo, W.-K. Ma, A.~M.-C. So, Y.~Ye, and S.~Zhang, ``Semidefinite
  relaxation of quadratic optimization problems,'' \emph{IEEE Signal Processing
  Magazine}, vol.~27, no.~3, pp. 20--34, 2010.

\bibitem{man2010probabilistic}
A.~Man-Cho~So, ``Probabilistic analysis of the semidefinite relaxation detector
  in digital communications,'' in \emph{Proceedings of the twenty-first annual
  ACM-SIAM symposium on Discrete Algorithms}.\hskip 1em plus 0.5em minus
  0.4em\relax Society for Industrial and Applied Mathematics, 2010, pp.
  698--711.

\bibitem{so2008unified}
A.~M.-C. So, Y.~Ye, and J.~Zhang, ``A unified theorem on {SDP} rank
  reduction,'' \emph{Mathematics of Operations Research}, vol.~33, no.~4, pp.
  910--920, 2008.

\bibitem{horn2012matrix}
R.~A. Horn and C.~R. Johnson, \emph{Matrix analysis}.\hskip 1em plus 0.5em
  minus 0.4em\relax Cambridge University Press, 2012.

\bibitem{bechar2009bernstein}
I.~Bechar, ``A bernstein-type inequality for stochastic processes of quadratic
  forms of gaussian variables,'' \emph{arXiv preprint arXiv:0909.3595}, 2009.

\bibitem{irmer2011coordinated}
R.~Irmer, H.~Droste, P.~Marsch, M.~Grieger, G.~Fettweis, S.~Brueck, H.-P.
  Mayer, L.~Thiele, and V.~Jungnickel, ``Coordinated multipoint: Concepts,
  performance, and field trial results,'' \emph{IEEE Communications Magazine},
  vol.~49, no.~2, pp. 102--111, 2011.

\bibitem{dinis2004class}
R.~Dinis and A.~Gusmao, ``{A class of nonlinear signal-processing schemes for
  bandwidth-efficient OFDM transmission with low envelope fluctuation},''
  \emph{IEEE Transactions on Communications}, vol.~52, no.~11, pp. 2009--2018,
  2004.

\bibitem{zhang2004matrix}
X.~Zhang, ``Matrix analysis and applications,'' \emph{Tsinghua and Springer
  Publishing House, Beijing}, pp. 71--100, 2004.

\end{thebibliography}

\end{document}